\documentclass[11pt,a4paper,oneside]{article}
\normalbaselines 
\usepackage{amsmath}
\usepackage{amssymb}
\usepackage{mathtools}
\usepackage{mathcomp}
\usepackage{textcomp}
\usepackage{amsthm}
\usepackage{enumerate}
\usepackage[auth-lg]{authblk}
\usepackage{cite}
\usepackage{tikz}
\usepackage{caption}
\usepackage{subcaption}
\usepackage[ruled,linesnumbered]{algorithm2e}
\usetikzlibrary{arrows}
\usetikzlibrary{shapes.geometric,patterns}
\usepackage{multicol}
\usepackage{tikz-qtree}
\usepackage{rotating}
\usepackage{url}

\newtheorem{theorem}{Theorem}\theoremstyle{plain}
\theoremstyle{plain}
\theoremstyle{plain}
\newtheorem{lemma}[theorem]{Lemma}\theoremstyle{plain}
\theoremstyle{plain}
\theoremstyle{plain}
\theoremstyle{plain}
\theoremstyle{plain}
\theoremstyle{plain}

\begin{document}
\title{\bf A Lex-BFS-based recognition algorithm for Robinsonian matrices}
\author[1,2]{Monique Laurent}
\author[1]{Matteo Seminaroti}
\date{}

\affil[1]{\small Centrum Wiskunde \& Informatica (CWI), Science Park 123, 1098 XG Amsterdam, The Netherlands}
\affil[2]{\small Tilburg University, P.O. Box 90153, 5000 LE Tilburg, The Netherlands}
\footnotetext{Correspondence to : \texttt{M.Seminaroti@cwi.nl} (M.~Seminaroti), \texttt{M.Laurent@cwi.nl} (M.~Laurent), CWI, Postbus 94079, 1090 GB, Amsterdam. Tel.:+31 (0)20 592 4386.}

\maketitle

\newcommand{\MS}{{\mathcal S}}
\newcommand{\oR}{{\mathbb R}}
\newcommand{\sfT}{{\sf T}}
\newcommand{\ignore}[1]{}
\newcommand{\QAP}{\text{\rm QAP}}
\newcommand{\MP}{\mathcal{P}}

\begin{abstract}
Robinsonian matrices arise in the classical seriation problem and play an important role in many applications where unsorted similarity (or dissimilarity) information must be reordered.
We present a new polynomial time algorithm to recognize Robinsonian matrices based on a new characterization of Robinsonian matrices in terms of straight enumerations of unit interval graphs. 
The algorithm is simple and is based essentially on lexicographic breadth-first search (Lex-BFS), using a divide-and-conquer strategy. 
When applied to a nonnegative symmetric $n\times n$ matrix with~$m$ nonzero entries and given as a weighted adjacency list, it runs in $O(d(n+m))$ time,  where $d$ is the depth of the recursion tree, which is at most the number of distinct nonzero entries of $A$.\\

\end{abstract}

\section{Introduction} 

An important question in  many classification problems is to  find an order of
 a collection of  objects   respecting some  given information about their pairwise (dis)similarities. 
The classic {seriation problem},  introduced by Robinson~\cite{Robinson51} for chronological dating, asks to order objects in such a way  that similar objects are ordered close to each other, and it has many applications in different fields (see~\cite{Innar10} and references therein). 

A symmetric matrix $A=(A_{ij})_{i,j=1}^n$ is  a {\em Robinson similarity} matrix  if its entries decrease monotonically in the rows and columns when moving away from the main diagonal, i.e., if 
$A_{ik}\le \min\{A_{ij},A_{jk}\}$ 
for all $1\le i\le j\le k\le n$.
Given a set of $n$ objects to order and a symmetric matrix $A=(A_{ij})$ which represents their pairwise correlations,
 the seriation problem asks to find (if it exists) a permutation $\pi$ of $[n]$ so that the permuted matrix $A_{\pi}=(A_{\pi(i)\pi(j)})$ is a Robinson matrix.
If such a permutation exists then $A$ is said to be a {\em Robinsonian similarity},
otherwise we say that data is affected by noise.
The definitions extend to dissimilarity matrices:  $A$ is a Robinson(ian) dissimilarity preciely when $-A$ is a Robinson(ian) similarity.
Hence  results can be directly transferred from one class to the other one.

\medskip
Robinsonian matrices play an important role in several hard combinatorial optimization problems and recognition algorithms are important in designing heuristic and approximation algorithms when the Robinsonian property is desired but the data is affected by noise (see e.g. \cite{Chepoi11,Fogel13,Laurent14}).
In the last decades, different characterizations of Robinsonian matrices have appeared in the literature, leading to different polynomial time recognition algorithms. Most characterizations are in terms of interval (hyper)graphs.

A graph $G=(V,E)$ is an {\em interval graph} if  its nodes can be labeled by intervals of the real line so that adjacent nodes correspond to intersecting intervals. 
Interval graphs arise frequently in applications and have been studied extensively in relation to hard optimization problems (see e.g. \cite{Bodlaender99,Cohen06,Mahesh91}).
A binary matrix has the {\em consecutive ones property (C1P)}  if its columns can be reordered in such a way that the ones are consecutive in each row.
Then a graph $G$ is an interval graph if and only if  its vertex-clique incidence matrix has C1P, where the rows are indexed by the vertices and the columns by the maximal cliques of $G$
\cite{Fulkerson65}.
 
A hypergraph $H=(V,\mathcal E)$ is a generalization of  the notion of  graph where  elements of $\mathcal E$, called  {\em hyperedges}, are subsets of $V$. The incidence matrix of $H$  is the $0/1$ matrix  whose rows and columns are labeled, respectively,  by the hyperedges and the vertices and with an entry 1 when the corresponding hyperedge contains  the corresponding vertex.
Then $H$ is   an {\em interval hypergraph} if its incidence matrix has C1P, i.e., if its vertices can be ordered in such a way that  hyperedges are intervals.

Given a dissimilarity matrix $A \in \MS^n$ and a scalar $\alpha$, the {\em threshold graph} $G_{\alpha}=(V,E_\alpha)$ has   edge set $E_{\alpha}= \{\{x,y\} : A_{xy} \leq \alpha\}$ and,
for  $x\in V$,  the ball $B(x,\alpha):=\{y \in V: A_{xy} \leq \alpha\}$  consists of  $x$ and its neighbors 
 in  $G_{\alpha}$. 
Let $\mathcal B$ denote the collection of all the balls of $A$ and  $H_{\mathcal{B}}$ denote the corresponding  \textit{ball hypergraph}, with vertex set  $V=[n]$ and with $\mathcal B$ as set of hyperedges. One can also build the intersection graph $G_{\mathcal B}$ of $\mathcal B$, where the balls are the vertices and connecting two vertices if the corresponding balls intersect.
Most of the existing algorithms are then based on the fact that a matrix $A$ is Robinsonian if and only if the ball hypergraph $H_{\mathcal{B}}$ is an interval hypergraph or, equivalently,  if the intersection graph $G_{\mathcal B}$ is an interval graph (see 
\cite{Mirkin84,Prea14}).

\medskip
Testing whether an $m\times n$  binary matrix with $f$ ones has C1P can be done in linear  time $O(n+m+f)$  (see  the first algorithm of Booth and Leuker \cite{Booth76} based on PQ-trees,  the survey  \cite{Dom09} and further references therein). 
Mirkin and Rodin \cite{Mirkin84} gave  the first polynomial algorithm to recognize Robinsonian matrices, with $O(n^4)$ running time, based on checking whether the ball hypergraph is an interval hypergraph and using the PQ-tree algorithm of \cite{Booth76} to check whether the incidence matrix has C1P. 
Later, Chepoi and Fichet \cite{Chepoi97} introduced a simpler algorithm that, using a divide-an-conquer strategy and sorting the entries of $A$, improved the running time to $O(n^3)$.
The same sorting preprocessing was used by Seston \cite{Seston08}, who improved the algorithm to $O(n^2\log n)$ by constructing paths in the threshold graphs of $A$.
Very recently, Pr\'ea and Fortin \cite{Prea14} presented a more sophisticated $O(n^2)$  algorithm, which uses the fact that the maximal cliques of the graph $G_{\mathcal B}$ are in one-to-one correspondence with the row/column indices of $A$.
Roughly speaking, they use the algorithm from Booth and Leuker~\cite{Booth76} to compute a first PQ-tree which they update throughout the algorithm.

A numerical spectral algorithm was introduced earlier by Atkins et al. \cite{Atkins98} for checking whether a similarity matrix $A$ is Robinsonian, based on reordering the entries of the Fiedler eigenvector of the Laplacian matrix associated to $A$, and it runs in $O(n(T(n)+n\log n))$ time, where $T(n)$ is the complexity of computing (approximately) the eigenvalues of an $n\times n$ symmetric matrix.

\medskip
In this paper we introduce a new combinatorial algorithm to recognize Robinsonian  matrices, based  
on characterizing them in terms of straight enumerations of unit interval graphs.
Unit interval graphs are a subclass of interval graphs, where the intervals labeling the vertices are required to have unit length. 
As is well known, they can be recognized in linear time $O(|V|+|E|)$ (see e.g. \cite{Corneil04,BangJensen07}~and references therein).
Many of the existing algorithms are based on the equivalence between unit interval graphs  and proper interval graphs (where the intervals should be pairwise incomparable) 
(see \cite{Roberts69,Roberts78}).
Unit  interval graphs have been recently characterized in terms of {\em straight enumerations}, which   are special orderings of the  classes of the `undistinguishability' equivalence relation, calling two vertices undistinguishable if they have the same closed neighborhoods (see \cite{Corneil95}).  This leads to alternative unit interval graph recognition algorithms (see \cite{Corneil95,Corneil04}), which we will use as main building block in our algorithm.
Our algorithm relies indeed on the fact that a similarity matrix $A$ is Robinsonian if and only if its level graphs (the analogues for similarities of the threshold graphs for dissimilarities) admit pairwise compatible straight enumerations (see Theorem \ref{thm: Robinsonian decomposition in uig}).

Our approach differs from the existing ones in the sense that it is not directly related  to interval (hyper)graphs, but it relies only on unit interval graphs (which are a simpler graph class than interval graphs) and on their straight enumerations.
Furthermore, our algorithm  does not rely on any sophisticated external algorithm such as the Booth and Leuker algorithm for C1P and no preprocessing to order the data is needed.
In fact, the most difficult task carried out by our algorithm  is a Lexicographic Breadth-First Search (abbreviated Lex-BFS), which is a variant of the classic Breadth-First Search (BFS), where the ties in the search are broken by giving preference to those vertices whose neighbors have been visited earliest (see \cite{Rose75} and \cite{Habib00}).
Following  \cite{Corneil04}, we in fact use the variant Lex-BFS+ introduced by \cite{Simon91} to compute straight enumerations.
Our algorithm uses a divide-and-conquer strategy with a merging step, tailored  to efficiently exploit the possible sparsity structure of the given similarity matrix $A$. 
Assuming the matrix $A$ is given as an adjacency list of an undirected weighted graph, our algorithm runs in $O(d(m+n))$ time, where $n$ is the size of $A$, $m$ is the number of nonzero entries of $A$ and $d$ is the depth of the recursion tree computed by the algorithm, which is upper bounded by the number $L$ of distinct nonzero entries  of $A$ (see Theorem \ref{thm:Robinsonian_recognition_complexity}).
Furthermore, we can return all the permutations reordering $A$ as a Robinson matrix using a PQ-tree data structure on which we perform only a few simple operations (see Section~\ref{secpermutations}).

Our algorithm uncovers an interesting link between straight enumerations of unit interval graphs and Robinsonian matrices which, to the best of our knowledge, has not been made before.
Moreover it provides an answer to an  open question posed by  M. Habib at the PRIMA Conference in Shanghai in June 2013, who asked whether it is possible to use Lex-BFS+ to recognize Robinsonian matrices \cite{Habib13}.
Alternatively one could check whether  the incidence matrix~$M$ of the ball hypergraph of $A$ has C1P, 
using the Lex-BFS based algorithm of 
\cite{Habib00},  in  time $O(r+c+f)$ time if $M$ is $r\times c$ with $f$ ones.
As $r\le nL$, $c=n$ and $f\le Lm$,  the overall time complexity is $O(L(n+m))$.
Interestingly, this approach is not mentioned by Habib. In comparison, an advantage of our approach is that it 
exploits the sparsity structure of the matrix $A$, as $d$ can be smaller than $L$.

This paper is an extended version of the work \cite{Laurent15}, which appeared in the proceedings of the 9th International Conference on Algorithms and Complexity (CIAC 2015) .

\subsubsection*{Contents of the paper}
Section \ref{secpreliminaries} contains preliminaries about  weak linear orders, straight enumerations and unit interval graphs.
In Section \ref{secRobinson} we characterize Robinsonian matrices in terms of straight enumerations of unit interval graphs.
In Section \ref{secalgmain} we introduce our recursive algorithm to recognize Robinsonian matrices, and then we discuss the complexity issues and explain how to return all the permutations reordering a given similarity matrix as a Robinson matrix.
The final Section \ref{secfinal} contains  some questions  for possible future work.

\section{Preliminaries}\label{secpreliminaries}

Throughout  $\MS^n$ denotes the set of symmetric $n\times n$ matrices. 
Given  a permutation $\pi$ of $[n]$ and a matrix $A\in~\mathcal~S^n$,  $A_\pi:=
(A_{\pi(i)\pi(j)})_{i,j=1}^n\in \MS^n$ is the matrix obtained by permuting  both the rows and columns of $A$  simultaneously according to $\pi$. For  $U\subseteq [n]$, $A[U]=(A_{ij})_{i,j\in U}$ is the principal submatrix of $A$ indexed by $U$. As we deal  exclusively with Robinson(ian) similarities, when speaking of a Robinson(ian) matrix, we mean a Robinson(ian)  similarity matrix.

\medskip
An ordered partition $(B_1,\ldots,B_p)$ of a finite set $V$ corresponds to  a {\em weak linear order} $\psi$ on $V$ (and vice versa), by setting
$x=_\psi y$ if $x,y$ belong to the same class $B_i$,  and 
$x <_{\psi} y$ if   $ x \in B_i$ and $ y \in B_j$ with $i<j$.
Then we also use the notation $\psi=(B_1,\ldots,B_p)$ and  $B_1<_\psi \ldots <_\psi B_p$. 
When all classes $B_i$ are singletons then $\psi$ is a  linear order (i.e., total order) of $V$.

The {\em reversal} of $\psi$ is the weak linear order, denoted $\overline \psi$, of the reversed ordered partition $(B_p,\ldots,B_1)$.
For  $U\subseteq V$, $\psi[U]$ denotes the {\em  restriction}  of the weak linear order $\psi$ to $U$.
Given disjoint subsets $U,W\subseteq V$, we say  $U\le_\psi W$ if $x\le_\psi y$ for all $x\in U,$ $y\in W$.
If $\psi_1$ and $\psi_2$ are weak linear orders on disjoint sets $V_1$ and $V_2$, then $\psi= (\psi_1,\psi_2)$ denotes their  {\em concatenation} which is a  weak  linear order  on $V_1\cup V_2$.

The following notions of  compatibility and refinement will play an important role in our treatment.
Two weak linear orders $\psi_1$ and $\psi_2$ on the same set $V$ are said to be {\em compatible} if 
there do not exist elements $x,y\in V$ such that $x<_{\psi_1} y$ and $y<_{\psi_2} x$.
Hence, $\psi_1$ and $\psi_2$ are compatible if and only if there exists a linear order $\pi$ of $V$ which is compatible with both $\psi_1$ and $\psi_2$
Then their {\em common refinement} is the weak linear order $\Psi=\psi_1\wedge \psi_2$ on $V$ defined by $x=_{\Psi} y$ if $x=_{\psi_\ell} y$ for all $\ell\in \{1,2\}$, and
$x<_{\Psi} y$ if $x\le_{\psi_\ell} y$ for all $\ell\in\{1,2\}$  with at least one strict inequality.

We will use the following fact, whose easy proof is omitted.
\begin{lemma}\label{thm: easy common refinement and pi}
Let $\psi_1,\dots,\psi_L$ be weak linear orders on $V$. Hence, $\psi_1,\dots,\psi_L$ are pairwise compatible if and only if there exists a linear order $\pi$ of $V$ which is compatible with each of $\psi_1,\dots,\psi_L$, in which case $\pi$ is compatible with their common refinement $\psi_1 \wedge \dots \wedge\psi_L$.
\end{lemma}

In what follows $V=[n]=\{1,\ldots,n\}$ is the vertex set of a graph $G=(V,E)$,  whose edges  are  pairs $\{x,y\}$ of distinct vertices $x,y\in V$.
For  $x\in V$, we denote by $N(x)=\{y\in V: \{x,y\}\in E\}$ the \textit{neighborhood} of $x$. Then, its  {\em closed  neighborhood} is the set $N[x] =  \{x\}\cup N(x)$. Two vertices $x,y\in V$ are {\em undistinguishable} if $N[x]=N[y]$. This defines an equivalence relation on $V$, whose classes are called the {\em blocks} of $G$.
Clearly, each block   is a clique of $G$. 
Two distinct blocks $B$ and $B'$  are said to be {\em adjacent}  if there exist two vertices $x \in B,$ $ y \in B'$ that are adjacent in $G$ or, equivalently, if $B\cup B'$ is a clique of $G$.
A {\em straight enumeration}  of $G$ is then a linear order $\phi=(B_1,\dots,B_p)$  of the blocks of $G$ such that, for any block $B_i$, the block $B_i$ and the blocks $B_j$ adjacent to it are consecutive in the linear order (see \cite{Hell02}).
The blocks $B_1$ and $B_p$ are called the {\em end blocks} of $\phi$ and $B_i$ (with $1<i<p$) are its  {\em inner blocks}.
Having a straight enumeration is a strong property, and not all graphs have one. In fact, this notion arises naturally in the context of unit interval graphs as recalled below.

A graph $G=(V=[n],E)$ is called an {\em interval graph} if its vertices can be mapped to intervals $I_1,\ldots,I_n$ of the real line such that, for distinct vertices $x,y\in V$, $\{x,y\}\in E$ if and only if $I_x\cap I_y\neq \emptyset$.
Such a set of intervals is  called a \textit{realization} of $G$, and it is not unique.
If the graph G admits a realization by unit intervals, then $G$ is said to be a \textit{unit interval graph}.

Interval graphs and unit interval graphs play an important role in many applications in different fields. Many NP-complete graph problems can be solved in polynomial time on interval graphs (this holds e.g. for the bandwidth problem \cite{Mahesh91}).
However, there are still problems which remains NP-hard also for interval graphs (this holds e.g. for the minimal linear arrangement problem \cite{Cohen06}).
It is well known that interval graphs and unit interval graphs can be recognized in $O(|V|+|E|)$ time \cite{Booth76,Simon91,Looges93,Corneil95,deFigueiredo95,Deng96,Corneil98, Habib00,Corneil04,BangJensen07}.
For a more complete overview on linear recognition algorithms for unit interval graphs, see \cite{Corneil04} and references therein.
Most of the above mentioned algorithms are based on the equivalence between unit interval graphs and proper interval graphs (i.e., graphs  admitting a realization  by pairwise incomparable intervals) or indifference graphs \cite{Roberts69}.
Furthermore, there exist several equivalent characterizations for unit interval graphs. 
The following one in terms of straight enumerations will play a central role in~our paper.

\begin{theorem}[\bf Unit interval graphs and straight enumerations] \label{thm:uig and straight enumeration} 
\cite{Deng96} A graph $G$ is a unit interval graph if and only if it has a straight enumeration. 
Moreover, if $G$ is  connected, then it has a unique (up to reversal) straight enumeration.
\end{theorem}

On the other hand, if $G$ is not connected, then any possible linear ordering of the connected components combined with any possible orientation of the straight enumeration of each connected component induces a straight enumeration of $G$.
The next theorem summarizes several known characterizations for unit interval graphs, combining results from   
\cite{Corneil95,Olariu91,Looges93,Roberts69,Roberts71,Gilmore64}.
Recall that  $K_{1,3}$ is the graph with one degree-3 vertex connected to three degree-1 vertices (also known as \textit{claw}).

\begin{theorem}\label{thm:unit graphs main theorem}
The following are equivalent for a graph $G=(V,E)$.
\begin{enumerate}
\item[(i)] $G$ is a unit interval graph.
\item[(ii)] $G$ is an interval graph with no induced subgraph $K_{1,3}$.
\item[(iii)] {\bf (3-vertex condition)} There is a linear  ordering $\pi$ of $V$ such that, for all $x,y,z\in V$,
\begin{equation}\label{eq3v}
 x<_\pi y<_\pi z, \ \{x,z\}  \in E \Longrightarrow \{x,y\},\{y,z\}  \in E.
 \end{equation}
\item[(iv)] {\bf (Neighborhood condition)} There is a linear  ordering $\pi$ of $V$ such that  for any $x\in V$ the vertices in  $N[x]$ are consecutive  with respect to~$\pi$.
\item[(v)] {\bf (Clique condition)} There is a linear ordering $\pi$ of $V$ such that the vertices contained in any maximal clique of~$G$ are consecutive with respect to $\pi$.
\end{enumerate}
\end{theorem}

\section{Robinsonian matrices and unit interval graphs}\label{secRobinson}

In this section we characterize Robinsonian matrices in terms of straight enumerations of unit interval graphs. 
We focus first on binary Robinsonian matrices.
We may view any symmetric binary matrix with all diagonal entries equal to 1  as the {\em extended} adjacency matrix of a graph.
The equivalence between binary Robinsonian matrices and indifference graphs (and thus with unit interval graphs) was first shown by Roberts \cite{Roberts69}.
Furthermore, as observed, e.g., by Corneil et al. \cite{Corneil95}, the ``neighborhood  condition"  for a graph is equivalent to its extended adjacency matrix having C1P. 
Hence we have  the following equivalence between Robinsonian binary matrices and unit interval graphs, which also follows as a direct application of Theorem \ref{thm:unit graphs main theorem}\textit{(iii)}.

\begin{lemma}\label{thm:Robinsonian matrices and 3-vertex condition}
Let $G=(V,E)$ be a graph and  $A_G$ be its extended adjacency matrix. Then,
$A_G$ is a Robinsonian similarity if and only if  $G$ is a unit interval graph.
\end{lemma}

The next result  characterizes  the linear orders that reorder the extended adjacency matrix $A_G$ 
as a Robinson matrix in terms of the straight enumerations of $G$. It is simple but will play a central role in our algorithm for recognizing Robinsonian similarities.

\begin{theorem}\label{thm:Robinson reordering and straight enumeration}
Let  $G=(V,E)$ be a 
graph. 
A linear order $\pi$ of $V$ reorders $A_G$ as a Robinson matrix if and only if there exists a straight enumeration of $G$ whose corresponding weak linear order $\psi$ is compatible with $\pi$, i.e., satisfies:
\begin{equation}\label{eqequi}
\forall x, y \in V\  \text{  with } \ x\neq_\psi y \qquad x<_{\pi} y \quad \Longleftrightarrow \quad x < _{\psi} y.
\end{equation}
\end{theorem}

\begin{proof}
Assume  that $\pi$ is a linear order of $V$ that reorders $A_G$ as a Robinson matrix.
Then it is easy to see that the 3-vertex condition  holds  for $\pi$ and that each block of $G$ is an interval w.r.t. $\pi$.
Therefore the order $\pi$ induces an order $\psi$ of the blocks: $B_1<_\psi \ldots <_\psi B_p$, with 
$B_i<_\psi B_j$ if and only if $x<_\pi y$ for all $x\in B_i$ and $y\in B_j$. In other words, $\psi$ is compatible with $\pi$ by construction.
Moreover,  $\psi$ defines a straight enumeration of $G$. Indeed, if $B_i<_\psi B_j <_\psi B_k$ and $B_i,B_k$ are adjacent then $B_j$ is adjacent to $B_i$ and $B_k$, since this property follows directly from  the 3-vertex condition for $\pi$. 

Conversely, assume that $B_1<_\psi \ldots <_\psi B_p$ is a straight enumeration of $G$ and let $\pi$ be a linear order of $V$ which is compatible with $\psi$, i.e., satisfies (\ref{eqequi}). We show that  $\pi$ reorders $A_G$ as a Robinson matrix. That is, we show that
if $x<_\pi y <_\pi z$,  then $(A_G)_{xz} \leq \min \{(A_G)_{xy}, (A_G)_{yz} \}$ or, equivalently, 
that $\{x,z\} \in E$ implies $\{x,y\}, \{y,z\}\in E$. If $x,z$ belong to the same block $B_i$ then $y\in B_i$ (using (\ref{eqequi})) and thus $\{x,y\},\{y,z\}\in E$ since $B_i$ is a clique.
Assume now that $x\in B_i$, $z\in B_k$ and $\{x,z\}\in E$. Then, $B_i<_\psi B_k$ and  $B_i$, $B_k$ are adjacent blocks and thus $B_i\cup B_k$ is a clique.
If $y\in B_i$ then $y$ is adjacent to $x$ and $z$ (since $B_i\cup B_k$ is a clique).
Analogously if $y\in B_k$. Suppose now that $y\in B_j$. Then, using (\ref{eqequi}), we have that $B_i<_\psi B_j<_\psi B_k$.
As $\psi$ is a straight enumeration with $B_i$,$B_k$ adjacent it follows that $B_j$ is adjacent to $B_i$ and to $B_k$ and thus $y$ is adjacent to $x$ and $z$.
\end{proof}

Hence,  in order to find the permutations reordering a given binary matrix $A$ as a Robinson matrix, it suffices to find all the possible straight enumerations of the corresponding graph $G$.
As is shown e.g. in \cite{Corneil95,Deng96}, this is a simple task and can be done in linear time. 
This is coherent with the fact that C1P  can be checked in linear time (see \cite{Dom09} and references therein). 

\medskip
We now  consider a general (nonbinary)  matrix $A$.
We first introduce its  `level graphs',  the analogues for similarity matrices of the  threshold graphs for dissimilarities.
Let $\alpha_0<\alpha_1<\dots<\alpha_L$ denote the distinct values taken by the entries of $A$.
The graph $G^{(\ell)}=(V,E_{\ell})$,  whose edges are the pairs $\{x,y\}$ with $A_{xy}\ge \alpha_\ell$,  is called the {\em $\ell$-th level graph} of~$A$. 
Let $J$ be the all ones matrix. Clearly, both $J$ and $-J$ are Robinson matrices. 
Hence, we may and will assume, without loss of generality, that $\alpha_{0}=0$. Then, $A$ is nonnegative and $G^{(1)}$ is its support graph. 
The level graphs can be used to decompose $A$ as a conic combination of binary matrices and, as already observed by Roberts \cite{Roberts78}, $A$ is Robinson precisely when these binary matrices are Robinson.
This is summarized in the next lemma, whose easy proof is omitted.

\begin{lemma}\label{thm:level graphs}
Let $A \in \mathcal{S}^n$ with distinct values $\alpha_0<\alpha_1<\dots<\alpha_L$ and with level graphs $G^{(1)},\ldots,G^{(L)}$. 
Then:
$$A=\alpha_0J+\sum_{\ell=1}^L{\left( \alpha_{\ell}-\alpha_{\ell-1}\right)A_{G^{(\ell)}}}.$$
Moreover, $A$ is Robinson if and only if $A_{G^{(\ell)}}$ is Robinson for each $\ell \in [L]$.
\end{lemma}

Clearly, if $A$ is a Robinsonian matrix then the adjacency matrices of its level graphs $G^{(\ell)}$~($\ell\in~[L]$) are Robinsonian too.
However,  the converse  is not true: it is easy to build a small example where $A$ is not Robinsonian although the extended adjacency matrix of each of its level graphs is Robinsonian. 
The difficulty lies in the fact that one needs to find a permutation that reorders {\em simultaneously}  the extended adjacency matrices of all the level graphs as  Robinson matrices.
Roberts \cite{Roberts78} first introduced a characterization of Robinsonian matrices in terms of indifference graphs (i.e. unit  interval graphs). 
Rephrasing his result using the notion of level graphs, he showed that $A$ is Robinsonian if and only if its level graphs have vertex  linear orders that are compatible (see \cite[Theorem 4.4]{Roberts78}). However, he does not give any algorithmic insight on how to find such orders.

Combining the  links between binary Robinsonian matrices and unit interval graphs (Lemma~\ref{thm:Robinsonian matrices and 3-vertex condition}) and  between  reorderings of binary Robinsonian matrices and straight enumerations of unit interval graphs  (Theorem \ref{thm:Robinson reordering and straight enumeration}) together with the decomposition result of Lemma \ref{thm:level graphs}, we obtain the following characterization of Robinsonian matrices.

\begin{theorem}\label{thm: Robinsonian decomposition in uig}
Let $A\in \MS^n$ with level graphs $G^{(1)},\ldots,G^{(L)}$. Then:
\begin{itemize}
\item[(i)]  $A$ is a Robinsonian  matrix if and only if there exist straight enumerations of  $G^{(1)},$ $\ldots,$ $ G^{(L)}$ whose corresponding weak linear orders $\psi_1,\ldots,\psi_L$ are pairwise compatible.
\item[(ii)]  A linear order $\pi$ of $V$ reorders $A$ as a Robinson matrix if and only if there exist pairwise compatible straight enumerations of $G^{(1)},\ldots, G^{(L)}$, whose corresponding common refinement is compatible with~$\pi$.
\end{itemize}
\end{theorem}

\begin{proof}
Observe first that if assertion \textit{(ii)} holds then \textit{(i)} follows directly using the result of Lemma \ref{thm: easy common refinement and pi}. We now prove \textit{(ii)}. Assume that $A$ is Robinsonian and let $\pi$ a linear order of $V$ that reorders $A$ as a Robinson matrix. Then $A_{\pi}$ is Robinson and thus, by lemma \ref{thm:level graphs}, each permuted matrix $(A_{G^{(\ell)}})_\pi$ is a Robinson matrix. 
Then, applying Theorem \ref{thm:Robinson reordering and straight enumeration}, for each $\ell \in [L]$, there exists a straight enumeration of $G^{(\ell)}$ whose corresponding weak linear ordering $\psi_\ell$ is compatible with~$\pi$.  
We can thus conclude that the common refinement of $\psi_1,\ldots,\psi_L$ is compatible in view of Lemma \ref{thm: easy common refinement and pi}.
Conversely, assume that there exist straight enumerations of $G^{(1)},\ldots, G^{(L)}$ whose corresponding weak linear orders $\psi_1,\ldots,\psi_L$ are pairwise compatible with $\pi$ and their common refinement is compatible with $\pi$. Then, by Theorem \ref{thm:Robinson reordering and straight enumeration}, $\pi$  reorders simultaneously each $A_{G^{(\ell)}}$ as a Robinson matrix and thus $A_{\pi}$ is Robinson, which shows that $A$ is Robinsonian. 
\end{proof}

\section{The algorithm}\label{secalgmain}
We describe here our algorithm for recognizing whether a given symmetric nonnegative matrix $A$ is Robinsonian. 
First, we introduce an algorithm which either returns a permutation reordering $A$ as a Robinson matrix or states that~$A$ is not a Robinsonian matrix. 
Then, we show how to modify it in order to return all the permutations reordering $A$ as a Robinson matrix.

\subsection{Overview of the algorithm}\label{secalg}
The algorithm is based on Theorem \ref{thm: Robinsonian decomposition in uig}. 
The main idea is to find straight enumerations of the level graphs of $A$ that are pairwise compatible and to compute their common refinement. 
The matrix $A$ is not Robinsonian precisely when these objects cannot be found.
As above, $L$ denotes the number of distinct nonzero entries of $A$ and throughout $G^{(\ell)}=(V=[n],E_\ell)$ is the $\ell$-th level graph, whose edges are the pairs $\{x,y\}$ with $A_{xy}\ge \alpha_\ell$, for $\ell\in [L]$. 

One of the main  tasks in the algorithm is to find (if it exists) a straight enumeration of a graph $G$ which is compatible with a given weak linear order $\psi$ of $V$.
Roughly speaking, $G$ will correspond to a level graph $G^{(\ell)}$ of $A$ (in fact, to a connected component of it), while $\psi$ will correspond to the common refinement of the previous level graphs $G^{(1)},\dots,G^{(\ell-1)}$. 
Hence, looking for a straight enumeration of $G$ compatible with $\psi$ will correspond to looking for a straight enumeration of $G^{(\ell)}$ compatible with previously selected straight enumerations of the previous level graphs $G^{(1)},\dots,G^{(\ell-1)}$.

Since the straight enumerations of the level graphs might not be unique, it is important to choose, among all the possible straight enumerations, the ones that lead to a common refinement (if it exists). 

If $G$ is a connected unit interval graph, its straight enumeration $\phi$ is unique up to reversal (see Theorem~\ref{thm:uig and straight enumeration}). 
On the other hand, if $G$ is not connected then any possible ordering of the connected components induces a straight enumeration, obtained by concatenating straight enumerations of its connected components.
This freedom in choosing the straight enumerations of the components is crucial in order to return {\em all} the Robinson orderings of $A$, and it is taken care of in Section \ref{secpermutations} using PQ-trees.

As we will see in Section~4.1.4, the choice of a straight enumeration of $G$ compatible with $\psi$ reduces to correctly orient straight enumerations of the connected components of $G$.

There are three main subroutines in our algorithm: \textit{CO-Lex-BFS} (see Algorithm~\ref{alg:Lex-BFSweak}), a variation of  Lex-BFS, which  finds and orders the connected components of the level graphs; \textit{Straight\_enumeration} (see Algorithm~\ref{alg:straght_enumeration}), which computes the straight enumeration of a connected graph as in \cite{Corneil04}; \textit{Refine} (see Algorithm~\ref{alg:refine}), a variation of partition refinement, which finds the common refinement of two weak linear orders.
These subroutines are used in the recursive algorithm \textit{Robinson} (see Algorithm \ref{alg:Robinsonian recognition}).

\subsubsection{Component ordering}
Our first subroutine is  \textit{CO-Lex-BFS} (where CO stands for `Component Ordering') in Algorithm~\ref{alg:Lex-BFSweak}.
Given a graph $G=(V,E)$ and a weak linear order $\psi$ of $V$, it  detects the connected components of $G$ and  orders them in a compatible way with respect to $\psi$. According to Lemma \ref{thm:order connected components} below, this is possible if $G$ admits a straight enumeration compatible with $\psi$.
\begin{lemma}\label{thm:order connected components}
Consider a graph  $G=(V,E)$ and  a weak linear order $\psi$ of $V$. If $G$ has   a straight enumeration $\phi$   compatible with $\psi$ then there exists an ordering $V_1, \ldots,V_c$ of the connected components of $G$ which is compatible with $\psi$, i.e.,  such that $V_{1} \leq_{\psi} \ldots \le_\psi V_c$.
\end{lemma}

\begin{proof}
If  $V_1,\ldots,V_c$ is the ordering of the components of $G$ which is induced by the straight enumeration $\phi$, i.e., 
$V_1<_\phi\ldots<_\phi V_c$, then 
$V_1\le_\psi \ldots \le_\psi V_c$ as $\phi$ is compatible with $\psi$.
\end{proof}

\begin{algorithm}[!ht] 
\caption{\textit{CO-Lex-BFS}$(G,\psi)$}
\label{alg:Lex-BFSweak}
\SetKwInput {KwIn}{input}
\SetKwInput {KwOut}{output}
\KwIn{a graph $G=(V,E)$, a weak linear order $\psi=(B_1,\dots,B_p)$ of $V$}
\KwOut{a linear order $\sigma$ of $V$ and a linear order $(V_1,\dots,V_c)$ of the connected components of $G$ compatible with $\psi$ and $\sigma$,  or STOP (no such linear order of the components exists)}
\vspace{2ex}
mark all the vertices as unvisited\\
$\omega = 1$\\
$V_{\omega}, B_{\omega}^{\min}, B_{\omega}^{\max}=\emptyset$\\
Let $u$ be a vertex in $B_1$\\
$label(u)=|V|$\\
\ForEach{$v \in V \setminus u$}{
	$label(v)=\emptyset$
}
\For{$i=|V|,\dots,0$}{
	let $S$ be the set of unvisited vertices with lexicographically largest label \label{alg:sliceS}\\
	pick arbitrarily a  vertex $p$ in $S$ coming first in $\psi$ and mark it as visited\\
	$\sigma(p)= |V| - i + 1$ \\
	\If{$label(p)= \emptyset$ \textup{\textbf{or}} $i=0$}{
	   \If{ there exists a block $B$ of $\psi$ such that $B \nsubseteq V_{\omega}$ and $B_{\omega}^{\min} <_{\psi} B <_{\psi} B_{\omega}^{\max}$}{
	    \textbf{stop} \hfill (no ordering of components compatible with $\psi$ exists)
              }  
	    	       		\If{$\omega \geq 2$}{
	       			\eIf{$V_\omega\subseteq B^{\min}_{\omega-1}$} {\label{alg:line Q node}
						swap $V_{\omega}$ and $V_{\omega-1}$ and modify $\sigma$ accordingly\\
					}{
		  				\If{$B^{\min}_{\omega} <_\psi  B^{\max}_{\omega-1}$} {
		  					\textbf{stop} \hfill (no ordering of components compatible with $\psi$ exists)
		  				} 
					}
			}
		$\omega = \omega + 1$\\
		$V_{\omega} = \emptyset$\\	
	}
	$V_{\omega}=V_{\omega} \cup \{p\}$	\\
	$B_{\omega}^{\min}$ is the first block in $\psi$ which meets $V_{\omega}$\\
	$B_{\omega}^{\max}$ is the last block in $\psi$ which meets $V_{\omega}$\\
	\ForEach{unvisited vertex $w$ in $N(p)$}{
		append $i-1$ to $label(w)$\\	
	}

}
\Return $(V_1,\dots,V_c)$ and $\sigma$
\end{algorithm}

Algorithm~\ref{alg:Lex-BFSweak}  is based on the following observations. When the vertex $p$ in the set $S$ at line~\ref{alg:sliceS} (which represents the current set of unvisited vertices with a tie, known as {\em slice} in Lex-BFS) has label $\emptyset$, it means that $p$ is not contained in  the current component $V_{\omega}$, so a new component containing $p$ is opened.
Every time a connected component $V_\omega$ has been completed, 
we check if it can be ordered along the already detected components in a compatible way with~$\psi$.
We also do this for the last completed component $V_c$, at the last iteration $i=0$ at line 9 of Algorithm ~\ref{alg:Lex-BFSweak}.
Let $B_\omega^{\min}$ and~$B_\omega^{\max}$ denote respectively the first and the last blocks of $\psi$ intersecting $V_{\omega}$.
We distinguish two cases:
\begin{enumerate}
\item if $V_{\omega}$ meets more than one block of $\psi$ (i.e., if $B_\omega^{\min} <_\psi B_\omega^{\max}$), we check if all the inner blocks  between $B^{\min}_\omega$ and $B^{\max}_\omega$ are contained in $V_\omega$. 
If this is not  the case, then  the algorithm stops. 
Moreover the algorithm also stops if both $V_\omega$ and $V_{\omega -1}$ meet exactly the same  two blocks, i.e., $B^{\min}_\omega=B^{\min}_{\omega-1}$ and $B^{\max}_\omega=B^{\max}_{\omega-1}$. In both cases  it is indeed not possible to order the components in a compatible way with $\psi$.
\item if $V_{\omega}$ meets only one block $B_k$ of $\psi$ (i.e., $V_\omega \subseteq B_k$)  and if this block $B_k$ is the first block of  the previous connected component $V_{\omega-1}$ (i.e., $B_k=B_{\omega-1}^{\min}$), then we swap $V_{\omega-1}$ and $V_{\omega}$ in order to make the ordering of the components compatible with $\psi$. 
The ordering $\sigma$ is updated by setting, for each $v \in V_{\omega-1}$ its new  ordering as $\sigma(v)+|V_{\omega}|$ and for each $v \in V_{\omega}$ as $\sigma(v)-|V_{\omega-1}|$.
Observe that if we are in the case  when both $V_\omega$ and $V_{\omega-1}$ are contained in $B_k$, then we  do not need to do this swap, i.e., the two components $V_\omega$ and $V_{\omega-1}$ can be ordered arbitrarily. 
\end{enumerate}

The next lemma shows the correctness of Algorithm \ref{alg:Lex-BFSweak}.

\begin{lemma}\label{thm:connected components}
Let $G=(V,E)$ be a graph and let $\psi$ be a weak linear order of $V$.
\begin{enumerate}\item[(i)]
If Algorithm \ref{alg:Lex-BFSweak} successfully terminates then the returned order $V_1,\ldots,V_c$ of the components satisfies\\ $V_1\le_\psi~\ldots~\le_\psi~V_c$.
\item[(ii)]
If Algorithm \ref{alg:Lex-BFSweak} stops then no ordering of the components exists that is compatible with $\psi$.
\end{enumerate}
\end{lemma}

\begin{proof}
\textit{(i)} Assume first that Algorithm \ref{alg:Lex-BFSweak} successfully terminates and returns the linear ordering $V_1,\ldots,V_c$ 
of the components.
Suppose for contradiction that
$V_{\omega-1}\not\le_\psi V_\omega$ for some $\omega\in [c]$ with $\omega\ge 2$.
Then there exist $x\in V_{\omega-1}$ and $y\in V_\omega$ such that $y<_\psi x$.
Let $z$ be the first vertex selected in the component $V_{\omega-1}$. Then, $z\le_\psi y$ (for if not the algorithm would have selected $y$ before $z$ when opening the component $V_{\omega-1}$).
Let $\psi = (B_1, \dots, B_p)$ and denote by $B_{\omega}^{min}$ and $B_{\omega}^{max}$, respectively, the first and last blocks of $\psi$ meeting $V_{\omega}$ ($B_{\omega-1}^{min}$ and $B_{\omega-1}^{max}$ are analogously defined).
Say, $x\in B_j$, $y\in B_i$ so that $i<j$, and $z\in B_r$. As $z\le_\psi y$, we have $B_r \le _\psi B_i$.
Suppose first that $B_r<_\psi B_i$. Then, $B_i$ is an inner block between $B^{\min}_{\omega-1}$ and $B^{\max}_{\omega-1}$ which is not contained in $V_{\omega-1}$ (since $y\in B_i$), yielding a contradiction since the algorithm would have stopped when dealing with the component $V_{\omega-1}$.
Suppose now that $B_r=B_i$. If $\psi[V_\omega]$ has only one block $B$, then  $B\subseteq B_i=B^{\min}_{\omega-1}$ and then the algorithm would have swapped $V_\omega$ and $V_{\omega-1}$.
Hence $\psi[V_\omega]$ has at least two blocks and
$B^{\min}_\omega \le_\psi B_i <_\psi B_j \le_\psi B^{\max}_{\omega-1}$, which is again a contradiction since the algorithm would have stopped. \\
\textit{(ii)} Assume now that the algorithm stops after the completion of 
 the component $V_\omega$.
Then $\psi[V_\omega]$  has at least two blocks. Suppose first that the algorithm stops because $B^{\min}_\omega<_\psi B^{\max}_{\omega-1}$. Then clearly one cannot have $V_{\omega-1}<_\psi V_\omega$. We show that we also cannot have $V_\omega <_\psi V_{\omega-1}$. For this assume for contradiction that $V_\omega <_\psi V_{\omega-1}$. 
Let $y$ be the first selected vertex in $V_\omega$ and let $x$ be the first vertex selected in $V_{\omega-1}$.
Then, $y\in B^{\min}_\omega$, $x\le_\psi y$ (for if not the algorithm would have considered the component $V_\omega$ before $V_{\omega-1}$), and thus 
$B^{\min}_{\omega-1} \le_\psi B^{\min}_{\omega}$.
If $B^{\min}_{\omega-1} <_\psi B^{\min}_{\omega}$ then the algorithm would have stopped earlier when examining $V_{\omega-1}$, since $B^{\min}_{\omega-1} <_\psi B^{\min}_\omega <_\psi B^{\max}_{\omega-1}$ and $B^{\min}_\omega\not\subseteq V_{\omega-1}$.
Hence, we have $B^{\min}_{\omega-1} = B^{\min}_{\omega}$ and, as $\psi[V_\omega]$ has at least two blocks, 
there exists a vertex $z\in B^{\max}_\omega$ such that $x<_\psi z$, which contradicts $V_\omega <_\psi V_{\omega-1}$.
Suppose now that the algorithm stops because 
$B^{\min}_\omega<_\psi B<_\psi B^{\max}_\omega$ and $B\not\subseteq V_\omega$. Let $x\in B^{\min}_\omega, y\in B^{\max}_\omega$ and $z\in B\setminus V_{\omega}$, and say $z\in V_{\omega'}$. Then we cannot have $V_{\omega'}<_\psi V_\omega$ since $x<_\psi z$, and we also cannot have $V_\omega<_\psi V_{\omega'}$ since $z<_\psi y$. 
Hence the two components $V_\omega$ and $V_{\omega'}$ cannot be ordered compatibly with $\psi$ and this  concludes the proof.
\end{proof}

\subsubsection{Straight enumerations}\label{secenu}
Once the connected components of $G$ are ordered, we  need to compute a straight enumeration of each connected component $G[V_\omega]$.
We do this with the  routine \textit{Straight\_enumeration} appplied to ($G[V_\omega],\sigma_\omega$), where $\sigma_\omega$ is a suitable given order of $V_\omega$ (namely, $\sigma_\omega=\sigma[V_\omega]$, where $\sigma$ is the vertex order returned by \textit{CO-Lex-BFS}$(G,\psi)$).
This routine is essentially the 3-sweep unit interval graph recognition algorithm of Corneil \cite{Corneil04} which, briefly, computes three times a Lex-BFS (each is named a \textit{sweep}) and uses the vertex ordering coming from the previous sweep to break ties in the search for the next sweep.
The only difference of \textit{Straight\_enumeration}$(G[V_\omega],\sigma_\omega)$ with respect to Corneil's algorithm is that we save the first sweep, because we use the order $\sigma_\omega$ returned by \textit{CO-Lex-BFS}.
We now describe the routine \textit{Straight\_enumeration} which is based on the algorithms of  \cite[\S 3]{Corneil95} and  \cite[\S 2]{Corneil04}. Below, $\deg_G(v)$ denotes the degree of the vertex $v$ in $G$.

\begin{algorithm}[!ht] 
\caption{\textit{Straight\_enumeration}$(G,\sigma)$}
\label{alg:straght_enumeration}
\SetKwInput {KwIn}{input}
\SetKwInput {KwOut}{output}
\KwIn{a connected graph $G=(V,E)$ and a linear order $\sigma$ of $V$}
\KwOut{a straight enumeration $\phi$ of $G$, or STOP ($G$ is not a unit interval graph)}
\vspace{2ex}
$\sigma^+$ = \textit{Lex-BFS+}($G,\sigma$)\\
$\sigma^{++}$ = \textit{Lex-BFS+}($G,\sigma^+$)\\
$i = 0$ \hfill (index of the blocks of $\psi$)\\
$L = R = 0$ \hfill (dummy variables to record the current block $B_i$)\\
\For {$v = 1, \dots,|V|$}{
	$lmn(v) = \min \{u: u \in N[v]\}$ \hfill (leftmost vertex adjacent to $v$) \\
	$rmn(v) = \max \{u: u \in N[v]\}$ \hfill (rightmost vertex adjacent to $v$) \\
	\eIf{$rmn(v)-lmn(v) \neq \deg_G(v)$}{
		\textbf{stop } \hfill ($G$ is not a unit interval graph)
		}{
		\eIf {$lmn(v)=L$ \textup{\textbf{and}} $rmn(v)=R$}{
			$C_i = C_i \cup \{v\}.$	
		}{
			$L = lmn(v)$\\	
			$R = rmn(v)$\\
			$i = i +1$\\
			$C_i = \{v\}$	
		}
	}
}
\Return $\phi=(C_1,\dots,C_q)$ 
\end{algorithm}

Basically, after the last sweep of Lex-BFS, for each vertex $v$ we define the leftmost vertex $lmn(v)$ and the rightmost vertex $rmn(v)$, according to $\sigma^{++}$, that are adjacent to $v$.
Checking whether $rmn(v)-lmn(v) = \deg_G(v)$ corresponds exactly to checking whether   the neighborhood condition holds for node $v$.
The vertices with the same leftmost and rightmost vertex are then indistinguishable vertices, and they form a block of $G$. The order of the blocks  follows the vertex order $\sigma^{++}$.

\subsubsection{Refinement of weak linear orders}
Given two weak linear orders $\psi$ and $\phi$ on~$V$, our second  subroutine \textit{Refine} in Algorithm \ref{alg:refine} computes their common refinement $\Phi=\psi \wedge \phi$ (if it exists).

\begin{algorithm}[!ht] 
\caption{\textit{\textit{Refine}}$(\psi,\phi)$}
\label{alg:refine}
\SetKwInput {KwIn}{input}
\SetKwInput {KwOut}{output}
\KwIn{two weak linear orders $\psi=(B_1,\dots,B_p)$ and $\phi=(C_1,\dots,C_q)$ of $V$}
\KwOut{their common refinement $\Phi=\psi\wedge \phi$, or $\Phi=\emptyset$ ($\psi$ and $\phi$ are not compatible)}
\vspace{2ex}
$B^{\max}$ is the last block of $\psi$ meeting $C_1$\\
$\Phi = \emptyset$\\
\eIf {there exists a block $B$ of $\psi$ such that $B<_\psi B^{\max}$ and $B\not\subseteq C_1$}{
	\Return $\emptyset$ \hfill ($\psi$ and $\phi$ are not compatible) \label{alg:Refine line 1}}
	{$W = V \setminus C_1$\\
	\eIf{$W = \emptyset$}{
			$\Phi=(\psi[C_1])$
			}{			
				$\Phi=(\psi[C_1],\textit{Refine}(\psi[W], \phi[W])$)
		}	
	}
\eIf {$\Phi$ is a weak linear order of $V$}{
	\Return $\Phi$} 
	{\Return $\emptyset$ \hfill ($\psi$ and $\phi$ are not compatible) \label{alg:Refine line 2}}
\end{algorithm}

We show the correctness of Algorithm  \ref {alg:refine}.

\begin{lemma}\label{thm:compatible refine}
If Algorithm~\ref {alg:refine} returns a weak linear order $\Phi$ of $V$, then $\Phi$ is the common refinement of $\psi$ and $\phi$.
If Algorithm~\ref {alg:refine} returns $\Phi = \emptyset$, then~$\psi$ and~$\phi$ are not compatible.
\end{lemma}

\begin{proof}
The proof is by induction on the number $q$ of blocks of $\phi$.
If $q=1$ then $\phi=(V)$ is clearly compatible with $\psi$ and the algorithm returns $\Phi=\psi$ as desired. Assume now $q\ge 2$.
Let $W=V\setminus C_1$. Then we can apply the induction assumption to $\psi[W]$ and $\phi[W]$ (which has $q-1$ blocks).

Assume first that the algorithm returns $\Phi$ which is a weak linear order of $V$. We show that $\Phi=\psi\wedge \phi$, i.e., that the following holds for all  $x,y\in V$:
\begin{equation}\label{eq:Phi}
\begin{array}{l}
x=_\Phi y \Longleftrightarrow x=_\psi y \text{ and } x=_\phi y,\\
x<_\Phi y \Longleftrightarrow x\le_\psi y \text{ and } x\le_\phi y \text{ with at least one strict inequality}.
\end{array}
\end{equation}
If $x,y\in C_1$ then $x=_\phi y$ and (\ref{eq:Phi}) holds since $\Phi[C_1]=\psi[C_1]$.
If $x,y\in V\setminus C_1$, then~(\ref{eq:Phi}) holds by the induction assumption.
Suppose now $x\in C_1$ and $y\in V\setminus C_1$. Then $x<_\phi y$ and $x<_\Phi y$. 
We show that $x\le_\psi y$ holds.
For this let $B_i$ (resp.,~$B_j$) be the block of $\psi$ containing $x$ (resp., $y$).
Then $B_i\leq_\psi B^{\max}$ since $B_1$ meets~$C_1$ as $x\in C_1$. Moreover, $B^{\max}\le_\psi B_j$, which implies  $x\le_\psi y$. 
Indeed, if one would have $B_j<_\psi B^{\max}$, then we would have $\Phi = \emptyset$ (line \ref{alg:Refine line 1} in Algorithm~\ref{alg:refine}), since $B_j\not\subseteq C_1$ as $y\in B_j\setminus C_1$, and thus $\Phi$ would not be a weak linear order of~$V$.

Assume now that the returned $\Phi$ is not a weak linear order of $V$.
If $\Phi = \emptyset$ (line~\ref{alg:Refine line 1} in Algorithm~\ref{alg:refine}), then there is a block $B<_\psi B^{\max}$ such that $B\not\subseteq C_1$, and we can pick elements $x\in B\setminus C_1$ and $y\in C_1\cap B^{\max}$ so that  $y<_\phi x$ and $x<_\psi y$, which shows that $\psi$ and $\phi$ are not compatible.
If $\Phi$ is a weak linear order of a subset $U \subset V$ (line~\ref{alg:Refine line 2} in Algorithm~\ref{alg:refine}), then it means that the weak linear order returned by the recursive routine \textit{Refine}$(\psi[W], \phi[W])$ is not a weak linear order of $W$ (but of a subset) and thus, by the induction assumption, $\psi[W]$ and~$ \phi[W]$ are not compatible and thus $\psi$ and $\phi$ neither. This concludes the proof.
\end{proof}

\subsubsection{Main algorithm}

We can now describe our main algorithm \textit{Robinson}$(A,\psi)$.  Given a nonnegative matrix $A\in \mathcal S^n$ and a weak linear order $\psi$ of $V=[n]$, it either returns  a weak linear order $\Phi$ of $V$ compatible with $\psi$ and with straight enumerations of the level graphs of $A$, or it indicates that such $\Phi$ does not exist.
The idea behind the algorithm is the following. 
We use the subroutines \textit{CO-Lex-BFS} and \textit{Straight\_enumeration} to order the components and compute the straight enumerations of the level graphs  of $A$, and we refine them using the subroutine \textit{Refine}. 
However, instead of refining the level graphs one by one on the full set $V$, we use a recursive algorithm based on a divide-and-conquer strategy, which refines smaller and smaller subgraphs of  the level graphs obtained by restricting to the connected components and thus working independently with the corresponding principal submatrices of $A$.
In this way we work with smaller subproblems and one may also skip some level graphs (as some principal submatrices of $A$ may have fewer distinct nonzero entries). 
This recursive algorithm is Algorithm  \ref{alg:Robinsonian recognition} below.

\begin{algorithm}[!ht] 
\caption{\textit{\textit{Robinson}}$(A,\psi)$}
\label{alg:Robinsonian recognition}
\SetKwInput {KwIn}{input}
\SetKwInput {KwOut}{output}
\KwIn{a nonnegative matrix~$A~\in \mathcal{S}^{n}$ and a weak linear order $\psi$ of $V=[n]$}
\KwOut{a weak linear order $\Phi$ compatible with $\psi$ and with straight enumerations of all the level graphs of $A$, or STOP (such an order $\Phi$ does not exist)}
\vspace{2ex}
$G$ is the support of A\\
\textit{CO-Lex-BFS}$(G,\psi)$ returns a linear order $(V_1,\dots,V_c)$ of the connected components of $G$ compatible with $\psi$ (if it exists) and a vertex order $\sigma$\\
$\Phi = \emptyset$\\
\For {$\omega=1,\dots,c$}{
	$\phi_{\omega}$ = \textit{Straight\_enumeration}($G[V_{\omega}],\sigma[V_\omega]$) \hfill (if $G[V_{\omega}]$ is a unit int. graph)\\
	\If {$\Phi_{\omega}$ = \textit{Refine}$(\psi[V_{\omega}],\phi_{\omega}) = \emptyset$}
	{
		\If {$\Phi_{\omega}$ = \textit{Refine}$(\psi[V_{\omega}],\overline{\phi}_{\omega}) = \emptyset$}
		{
		\textbf{stop} \hfill (no straight enumeration compatible with $\psi[V_{\omega}]$ exists) \label{alg:Rob line}
		}
	}
	$a'_{\min}$ is the smallest nonzero entry of $A[V_{\omega}]$\\
	$A'[V_{\omega}]$ is obtained from $A[V_{\omega}]$ by setting entries with value $a'_{\min}$ to zero\\
	\eIf{$A'[V_{\omega}]$ is diagonal}{
		$\Phi = (\Phi,\Phi_{\omega})$\\				
	}{
		$\Phi = (\Phi,Robinson(A'[V_{\omega}],\Phi_{\omega}))$ \label{alg:Rob line 2}\\
	}
}
\Return: $\Phi$ 
\end{algorithm}

\medskip
The algorithm \textit{Robinson}$(A,\psi)$ works  as follows. We are given as input a symmetric nonnegative matrix $A \in \MS^n$ and a weak linear order $\psi$ of $V=[n]$. 
Let $G$ be the support of $A$. First, we find the connected components of $G$ and we order them in a compatible way with $\psi$. If this is not possible, then we stop as there do not exist straight enumerations of the level graphs of $A$ compatible with $\psi$ (Theorem~\ref{thm:connected components}).
Otherwise, we initialize the weak linear order $\Phi$, which at the end of the algorithm will represent a common refinement of the straight enumerations of the level graphs of $A$.
In order to find~$\Phi$, we divide the problem over the connected components of $G$. The idea is then to work independently on each connected component $V_{\omega}$ and to find its (unique up to reversal) straight enumeration $\phi_{\omega}$ and the common refinement $\Phi_{\omega}$ of $\psi[V_{\omega}]$ and~$\phi_\omega$.

For each component $V_{\omega}$, we compute the straight enumeration $\phi_{\omega}$ of $G[V_\omega]$ if it exists, else we stop (Theorem~\ref{thm:uig and straight enumeration}). 
Since $\phi_{\omega}$ is unique up to reversal, we check if either $\phi_{\omega}$ or $\overline{\phi}_{\omega}$ is compatible with $\psi[V_{\omega}]$.
Specifically, we first compute the common refinement $\Phi_{\omega}$ of $\psi[V_{\omega}]$ and $\phi_{\omega}$. 
If it is nonempty we continue (Lemma~\ref{thm:compatible refine}), while if it is is empty we compute the common refinement $\Phi_{\omega}$ of $\psi[V_{\omega}]$ and $\overline{\phi}_{\omega}$.
If such a common refinement is nonempty we continue~(Lemma~\ref{thm:compatible refine}), 
while if it is again empty this time we stop, as no straight enumeration of $G_{\omega}$ compatible with $\psi[V_{\omega}]$ exists.
Finally, we set to zero the smallest nonzero entries of $A[V_{\omega}]$, obtaining the new matrix $A'[V_{\omega}]$ (whose nonzero entries take fewer distinct values than the matrix $A[V_\omega]$).
Now, if the matrix $A'[V_{\omega}]$ is diagonal, then we concatenate $\Phi_{\omega}$ after $\Phi_{\omega-1}$ in $\Phi$. 
Otherwise, we make a recursive call, where the input of the recursive routine is the matrix $A'[V_{\omega}]$ and $\Phi_{\omega}$. 
If the algorithm successfully terminates, then the concatenation $(\phi_1,\dots,\phi_c)$ will represent a straight enumeration of $G$, and $\Phi=(\Phi_1,\dots,\Phi_c)$ will represent the common refinement of this straight enumeration with the given weak linear order $\psi$ and with the  level graphs of $A$.

The final algorithm is Algorithm \ref{alg:Robinsonian_main} below. 

\begin{algorithm}[!ht]
\caption{\textit{\textit{Robinsonian}}$(A)$}
\label{alg:Robinsonian_main}
\SetKwInput {KwIn}{input}
\SetKwInput {KwOut}{output}
\KwIn{a nonnegative  matrix $A \in \mathcal{S}^{n}$}
\KwOut{a permutation $\pi$ such that $A_{\pi}$ is Robinson or stating that $A$ is not Robinsonian}
\vspace{2ex}
$\psi=(V)$\\
\eIf{\textit{Robinson}($A,\psi$) stops}{
		``A is NOT Robinsonian"		
	}{	
		$\Phi$=\textit{Robinson}($A,\psi$)\\
		\Return: a linear order $\pi$ of $V$ compatible with $\Phi$;
}
\end{algorithm}

Roughly speaking, every time we make a recursive call, we are basically passing to the next level graph of $A$.  
Hence, each recursive call can be visualized as the node of a recursion tree, whose root is defined by the first recursion in Algorithm \ref{alg:Robinsonian_main}, and whose leaves (i.e. the pruned nodes) are the subproblems whose corresponding submatrices are diagonal.

The correctness of Algorithm \ref{alg:Robinsonian_main} follows directly from the correctness of 
Algorithm \ref{alg:Robinsonian recognition}, which is shown by the next theorem.
Indeed, assume that Algorithm \ref{alg:Robinsonian recognition} is correct. Then, if Algorithm \ref{alg:Robinsonian_main} terminates then it computes a weak linear order $\Phi$ 
compatible with straight enumerations of the level graphs of $A$ and thus the returned order $\pi$ orders $A$ as a Robinson matrix in view of Theorem \ref{thm: Robinsonian decomposition in uig} \textit{(ii)}. On the other hand, if Algorithm \ref{alg:Robinsonian_main} stops then Algorithm \ref{alg:Robinsonian recognition} stops with the input $(A,\psi=(V))$. Then no weak linear order $\Phi$ exists which is compatible with straight enumerations of the level graphs of $A$ and thus, in view of 
Theorem~\ref{thm: Robinsonian decomposition in uig} \textit{(i)}, $A$ is not Robinsonian.

\begin{theorem} \label{thm:correctness}
Consider a weak linear order $\psi$ of $V=[n]$ and a nonnegative matrix  $A\in \mathcal S^n$ ordered compatibly with~ $\psi$.
\begin{itemize}
\item[(i)] If Algorithm \ref{alg:Robinsonian recognition}  terminates,  then there exist straight enumerations $\phi^{(1)},\ldots, \phi^{(L)}$ of the level graphs $G^{(1)},\ldots,G^{(L)}$ of $A$ such that the  returned weak linear order $\Phi$ is compatible with each of them and with $\psi$.
\item[(ii)]
If Algorithm \ref{alg:Robinsonian recognition} stops then there do not exist straight enumerations of the level graphs of $A$ that are pairwise compatible and compatible with $\psi$.
\end{itemize}
\end{theorem}
\begin{proof}
The proof is by induction on the number $L$ of distinct nonzero entries of the matrix $A$. 
We first consider the base case $L=1$, i.e., when $A$ is (up to scaling) 0/1 valued. We first show \textit{(i)} and assume that the algorithm terminates successfully and returns $\Phi$. Then $G$ is the support of $A$,  \textit{CO-Lex-BFS}$(G,\psi)$ orders the components of $G$ as $V_1\le_\psi \ldots \le_\psi V_c$, and 
$\Phi=(\Phi_1,\ldots, \Phi_c)$ where each $\Phi_\omega=\Phi[V_\omega]$ is build as the common refinement of~$\psi[V_\omega]$ and a straight enumeration of~$G[V_\omega]$ (either $\phi_{\omega}$ or $\overline{\phi}_{\omega}$). Hence $G$ has a straight enumeration $\phi$ and the returned~$\Phi$ is compatible with $\phi$ and~$\psi$. 

We now show \textit{(ii)} and assume that Algorithm \ref{alg:Robinsonian recognition} stops. If it stops when applying \textit{CO-Lex-BFS}$(G,\psi)$, then no order of the components of $G$ exists that is compatible with $\psi$ and thus no straight enumeration of $G$ exists that is compatible with $\psi$ 
(Lemma~\ref{thm:order connected components}). 
If the algorithm stops when applying \textit{Straight\_enumeration} to $G[V_{\omega}]$ then no straight enumeration of $G[V_{\omega}]$ exists. Else, if the algorithm stops at line \ref{alg:Rob line} in Algorithm~\ref{alg:Robinsonian recognition}, then $\psi[V_{\omega}]$ is not compatible with neither $\phi_{\omega}$ nor $\overline{\phi}_{\omega}$.
Because  $G[V_{\omega}]$ is connected, $\phi_{\omega}$ and $\overline{\phi}_{\omega}$ are its unique straight enumerations (see Theorem~\ref{thm:uig and straight enumeration}) and therefore no straight enumeration of $G[V_{\omega}]$ is compatible with $\psi[V_{\omega}]$.
In both cases, no straight enumeration of $G$ exists that is compatible with~$\psi$.

\smallskip
We now assume that Theorem \ref{thm:correctness} holds for any matrix whose entries take  at most $L-1$ distinct nonzero values.
We show that the theorem holds when considering  $A$ whose nonzero entries take $L$ distinct values. We follow the same lines as the above proof for the case $L=1$, except that we use recursion for some components.
First, assume that the algorithm terminates and returns $\Phi$. Then, $\Phi=(\tilde \Phi_1,\ldots, \tilde \Phi_c)$ after ordering the components compatibly with $\psi$ as 
$V_1\le_\psi \ldots \le_\psi V_c$, constructing the common refinement $\Phi_{\omega}$ of $\psi[V_{\omega}]$ and a straight enumeration (say) $\phi_{\omega}$ of $G[V_{\omega}]$, and having  
$\tilde \Phi_\omega= \textit{Robinson}(A'[V_{\omega}], \Phi_\omega)$, where $A'[V_\omega]$ is obtained from $A[V_\omega]$ by setting to 0 its entries with smallest nonzero value. 
By the induction assumption, $\tilde \Phi_\omega$ is compatible  with straight enumerations of the level graphs of the matrix $A'[V_\omega]$ and with $\Phi_\omega$. 
As $\tilde \Phi_\omega$ is compatible with~$\Phi_\omega$, which refines both $\psi[V_{\omega}]$ and~$\phi_{\omega}$, it follows that 
$\tilde \Phi_\omega$ is compatible with~$\psi[V_\omega]$ and~$\phi_\omega$.
Therefore, $\tilde \Phi_\omega$ is compatible with straight enumerations of all the level graphs of~$A[V_\omega]$ and thus $\Phi=(\tilde \Phi_1,\ldots, \tilde \Phi_c)$ is compatible with~$\psi$ and all level graphs of~$A$, as desired. 

Assume now that the algorithm stops. 
If the algorithm stops at \textit{CO-Lex-BFS}$(G,\psi)$, then no linear order of the connected components of $G$ exists that is compatible with $\psi$ and then no straight enumeration of $G$ exists that is compatible with $\psi$ (Lemma~\ref{thm:order connected components}), giving the desired conclusion.
If the algorithm stops at line ~\ref{alg:Rob line}, then a connected component $V_{\omega}$ is found for which $\psi[V_{\omega}]$ is not compatible with any straight enumeration of $G[V_{\omega}]$, giving again the desired conclusion.

Assume now that the algorithm stops at line~\ref{alg:Rob line 2}, i.e., there is a component~$V_{\omega}$ for which the algorithm terminates when applying $\textit{Robinson}(A'[V_{\omega}], \Phi_\omega)$.
Then, by the induction assumption, we know that:
\begin{equation}\label{eq:induction correctness} \tag{*}
\begin{array}{l}
\text{no straight enumerations of the level graphs of } A'[V_{\omega}] \text{ exist  }\\
\text{that are pariwise compatible and compatible with  } \Phi_{\omega},
\end{array}
\end{equation}

\noindent
where $\Phi_{\omega}$ is the common refinement of~$\psi[V_{\omega}]$ and a straight enumeration (say)~$\phi_{\omega}$ of $G[V_{\omega}]$.
Assume, for the sake of contradiction, that there exist straight enumerations~$\varphi^{(1)},\dots,\varphi^{(L)}$ of the level graphs~$G^{(1)}=G,\dots,G^{(L)}$ of $A$, that are pairwise compatible and compatible with $\psi$.
In particular, $\varphi^{(1)}[V_{\omega}]$ is a straight enumeration of $G[V_{\omega}]$ compatible with $\psi[V_{\omega}]$.
If $\varphi^{(1)}[V_{\omega}] = \phi_{\omega}$, then the restrictions $\varphi^{(\ell)}$ ($\ell \geq 2$) yield straight enumerations of the level graphs of $A'[V_{\omega}]$ that are pairwise compatible and compatible with $\phi_{\omega}$ and $\psi[V_{\omega}]$, and thus with their refinement 
$\Phi_{\omega} = \psi[V_{\omega}] \wedge \phi_{\omega}$, contradicting~(\ref{eq:induction correctness}).
Hence, $\varphi^{(1)}[V_{\omega}] = \overline{\phi}_{\omega}$, so that $\psi[V_{\omega}]$ is compatible with both $\phi_{\omega}$ and its reversal $\overline{\phi}_{\omega}$.
This implies that $\psi[V_{\omega}] = (V_{\omega})$.
But then the reversals $\overline{\varphi}^{(2)}[V_{\omega}], \dots, \overline{\varphi}^{(L)}[V_{\omega}]$ provide straight enumerations of the level graphs of $A'[V_{\omega}]$ that are pairwise  compatible and compatible with $\overline{\varphi}^{(1)}[V_{\omega}] = \phi_{\omega} = \Phi_{\omega}$.
This contradicts again~(\ref{eq:induction correctness}) and concludes the proof.
\end{proof}

\subsection{Complexity analysis }\label{seccomplexity}
We now study the complexity of our main  algorithm. 
First we discuss the complexity of the two  subroutines \textit{CO-Lex-BFS} and \textit{Refine} in Algorithms 1 and 2 and then we  derive the complexity of the final Algorithm 4.
In the rest of the section, we let $m$ denote the number of nonzero (upper diagonal) entries of $A$, so that $m$ is the number of edges of the support graph $G=G^{(1)}$ and $m=|E_1| \ge |E_2|\ge \ldots \ge |E_L|$ for the level graphs of $A$.
We  assume that $A$ is a nonnegative symmetric matrix, which is  given as an adjacency list of an undirected weighted graph, where each vertex $x \in V$ is linked to the list of vertex/weight pairs corresponding to the neighbors $y$ of $x$ in $G$ with nonzero entry~$A_{xy}$.

A simple but important observation that we will repeatedly use is that, given a weak linear order $\psi$ of $V$, we can assume the vertices $V$ to be ordered according to a linear order~$\tau$ of $V$ compatible with~$\psi$. 
Then, the blocks of $\psi$ are intervals of the order $\tau$ and thus one can check whether a given set $C\subseteq V$ is contained in a block $B$ of $\psi$ in $O(|C|)$ operations (simply by comparing each element of $C$ to the end points of the interval $B$). 
Furthermore, the size of any block of $\psi$ is simply given by the difference between its extremities (plus one).

\begin{lemma}\label{thm:CO-Lex-BFS_complexity}
Algorithm \ref{alg:Lex-BFSweak} runs in $O(|V|+|E|)$ time.
\end{lemma}

\begin{proof}
It is well known that Lex-BFS can be implemented in linear time $O(|V|+|E|)$. In our implementation of Algorithm \ref{alg:Lex-BFSweak} we will follow  the linear time implementation of Corneil \cite{Corneil04}, which uses the data structure based on the paradigm of ``partitioning" presented in \cite{Habib00}.
Recall that the blocks of $\psi$ are intervals in $\tau$, which is a linear order compatible with~$\psi$.
In order to carry out the other operations about the components of $G$ we maintain a doubly linked list, where each node of the list represents a connected component $V_{\omega}$ of $G$ and it has a pointer to the connected component $V_{\omega-1}$ ordered immediately before $V_\omega$ and to the connected component $V_{\omega+1}$ ordered immediately after~$V_\omega$.
Then, swapping two connected components can be done simply by swapping the left and right pointers of the corresponding connected components in the doubly linked list.
Furthermore, each node in this list contains the set of vertices in $V_{\omega}$, the first block $B_{\omega}^{\min}$ and the last block $B_{\omega}^{\max}$ in $\psi$ meeting $V_{\omega}$.  These two blocks $B^{\min}_\omega$ and $B^{\max}_\omega$ can be found in time $O(|V_\omega|)$ as follows. First one finds the smallest element $v_{\min}$ (resp. the  largest  element $v_{\max}$) of $V_\omega$ in the order $\tau$, which can be done in $O(|V_\omega|)$.  Then, $B^{\min}_\omega$ is the block of $\psi$ containing $v_{\min}$, which can be found in $O(|V_\omega|)$. Analogously for $B^{\max}_\omega$, which is the block of $\psi$ containing $v_{\max}$.  Checking whether $V_\omega$ is contained in the block $B^{\min}_{\omega-1}$ can be done in $O(|V_\omega|)$ (since $B^{\min}_{\omega-1}$ is an interval). In order to check whether all the inner blocks between $B^{\min}_\omega$ and $B^{\max}_\omega$ are contained in $V_\omega$ we proceed as follows. Let $B_\omega$ be the union of these inner blocks, which is an interval of $\tau$. First we compute the sets $V_\omega\cap B^{\min}_\omega$ and $V_\omega\cap B^{\max}_\omega$, which can be done in $O(|V_\omega|)$. Then we need to check whether $B_\omega\subseteq V_\omega$ or, equivalently, whether 
the two sets $V_\omega \setminus (B^{\min}_\omega\cup B^{\max}_\omega)$ and  $B_\omega$ are equal. For this we check first whether 
$V_\omega \setminus (B^{\min}_\omega\cup B^{\max}_\omega)$ 
is contained in $B_\omega$ (in time $O(|V_\omega|)$) and then whether these two sets have the same cardinality, which can be done in $O(|V_\omega|)$. 
Hence, the complexity of this task is $O(\sum_\omega|V_\omega|)=O(|V|)$. 
Therefore we can conclude that the overall complexity of Algorithm \ref{alg:Lex-BFSweak} is $O(|V|+|E|)$.
\end{proof}

\begin{lemma}\label{thm:refine_complexity}
Algorithm \ref{alg:refine} runs in $O(|V|)$ time.
\end{lemma}

\begin{proof}
We show the lemma using induction on the number $q$ of blocks of $\phi$. 
Recall that the blocks of $\psi$ are intervals in~$\tau$, which is a linear order compatible with $\psi$. 
If $q=1$ the result is clear since the algorithm returns $\Phi=\phi$ without any work.
Assume $q\ge 2$. The first task is to compute the last block $B^{\max}$ of $\psi$ meeting $C_1$. For this, as in the proof of the previous lemma, one finds the largest element $v_{\max}$ of $C_1$ in the order $\tau$ and one returns the block of $\psi$ containing $v_{\max}$, which can be done in $O(|C_1|)$.
Then let $B$ be the union of the blocks preceding $B^{\max}$. In order to check whether $B\subseteq C_1$ or, equivalently, whether $C_1\setminus B^{\max}= B$, we proceed as in the previous lemma: we first check whether $C_1\setminus B^{\max}\subseteq  B$ and then whether $|C_1\setminus B^{\max}|= |B|$, which can be done in $O(|C_1|)$.
Hence, the running time is $O(|C_1|)$ for this task which, together with 
the running time $O(|V\setminus C_1|)$ for the recursive application of \textit{Refine} to the restrictions of $\psi$ and $\phi$ to
the set $V\setminus C_1$, gives an overall running time $O(|V|)$.
\end{proof}

We can now complete the complexity analysis of our algorithm.

\begin{theorem}\label{thm:Robinsonian_recognition_complexity}
Let $A$ be a nonnegative $n\times n$ symmetric matrix given as a weighted adjacency list and let $m$ be the number of (upper diagonal) nonzero entries of $A$. 
Algorithm~\ref{alg:Robinsonian_main} recognizes whether $A$ is a Robinsonian matrix in time $O(d(n+m))$, where $d$ is the depth of the recursion tree created by Algorithm \ref{alg:Robinsonian_main}. Moreover, $d\le L$, where $L$ is the number of distinct nonzero entries of $A$.
\end{theorem}
\begin{proof}
We show the result using induction on the depth $d$ of the recursion tree.
In Algorithm \ref{alg:Robinsonian_main} we are given a matrix $A$ and its support graph $G$, and we set $\psi=(V=[n])$.
First we run the routine \textit{CO-Lex-BFS}$(G,\psi)$ in $O(n+m)$ time, in order  to find and order the components of $G$.
For each component $V_\omega$, the following tasks are performed. 
We compute a straight enumeration $\phi_\omega$ of $G[V_\omega]$, in time $O(n_\omega+m_\omega)$ where $n_\omega= |V_\omega|$ and $m_\omega$ is the number of edges of $G[V_\omega]$.
The reversal $\overline{\phi}_{\omega}$ can be computed in $O(|V_{\omega}|)$ by simply reversing the ordered partition $\phi_{\omega}$, which is stored in a double linked list.
Hence, we apply the routine \textit{Refine} to~$\psi[V_\omega]$ and~$\phi_\omega$ (or $\overline{\phi}_\omega$), which can be done in $O(n_\omega)$ time.
Then we build the new matrix $A'[V_\omega]$ and checks whether it is diagonal, in time $O(m_\omega)$. 
Finally, by the induction assumption, the recursion step \textit{Robinson}$(A'[V_\omega],\Phi_\omega)$ is carried out in time 
$O(d_\omega(n_\omega+m_\omega))$, where $d_\omega $ denotes the depth of the corresponding recursion tree.
As $d_\omega\le d-1$ for each $\omega$, after summing up, we find that  the overall complexity is $O(d(n+m))$.
The last claim: $d\le L$ is clear since the number of distinct nonzero entries of the current matrix decreases by at least 1 at each recursion node. 
\end{proof}

\subsection{Finding all Robinsonian orderings}\label{secpermutations}

In general, there might exist several permutations reordering a given matrix $A$ as a Robinson matrix.
We show here how to return all Robinson orderings of a given matrix $A$, using the PQ-tree data structure of \cite{Booth76}.

A PQ-tree $\mathcal{T}$ is a special rooted ordered tree. The leaves are in  one-to-one correspondence with the elements of the groundset $V$ and their order gives a linear order of $V$.
The nodes of $\mathcal T$ can be of two types, depending on how their children can be ordered. Namely, 
for a \textit{P-node} (represented by a circle), its children may be arbitrary reordered; for a \textit{Q-node} (represented by a rectangle), only the order of its  children may be  reversed.
Moreover, every node has at least two children.
Given a  node $\alpha$ of $\mathcal T$,  $\mathcal{T}_{\alpha}$ denotes the subtree of $\mathcal T$ with root $\alpha$.

A straight enumeration $\psi=(B_1,\dots,B_p)$ of a graph $G=(V,E)$ corresponds in a unique way to a PQ-tree $\mathcal{T}$ as follows. 
If $G$ is connected, then the {root}  of $\mathcal T$ is a Q-node, denoted $\gamma$, and it has children $\beta_1,\ldots,\beta_p$ (in that order). For $i\in [p]$, the node $\beta_i$ is a P-node corresponding to the block $B_i$  and its children are the elements of the set $B_i$, which are the leaves of the subtree $\mathcal T_{\beta_j}$.
If a block $B_i$ is a singleton then no node $\beta_i$ appears and the element of $B_i$ is directly a child of the root $\gamma$ (see the example in Figure \ref{fig:example PQ-Tree contig}).

\begin{figure}[!h]
\centering
\begin{minipage}{.4\textwidth}
\centering
\begin{equation*}
A_G =
\bordermatrix{
~ & \textbf{1} & \textbf{2} & \textbf{3} & \textbf{4} & \textbf{5} & \textbf{6}\cr
\textbf{1} & 1 & 1 & 1 & 1 & 0 & 0 \cr
\textbf{2} & 1 & 1 & 1 & 1 & 1 & 1 \cr
\textbf{3} & 1 & 1 & 1 & 1 & 1 & 1 \cr
\textbf{4} & 1 & 1 & 1 & 1 & 1 & 1 \cr
\textbf{5} & 0 & 1 & 1 & 1 & 1 & 1 \cr
\textbf{6} & 0 & 1 & 1 & 1 & 1 & 1 \cr
}
\end{equation*} 
\end{minipage} 
\begin{minipage}{.55\textwidth}
\centering
\includegraphics{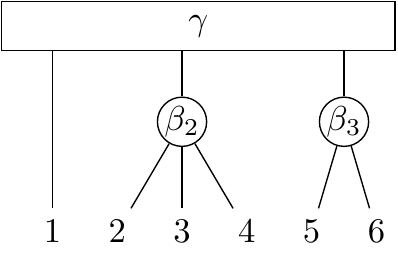}
\end{minipage}
\caption{A connected graph $G$ and the PQ-tree corresponding to its straight enumeration}
\label{fig:example PQ-Tree contig}
\end{figure}

If $G$ is not connected, let $V_1,\dots,V_c$ be its connected components. 
For each connected component $G[V_{\omega}]$, $\mathcal T_\omega$ is its PQ-tree (with root $\gamma_\omega$) as indicated above. 
Then, the full PQ-tree  $\mathcal{T}$ is obtained by inserting a P-node $\alpha$ as ancestor, whose children are the subtrees $\mathcal{T}_1,\ldots,\mathcal{T}_c$ (see Figure \ref{fig:example PQ-tree disconnected}).

\begin{figure}[ht!]
\centering
  	\includegraphics[width=.55\linewidth]{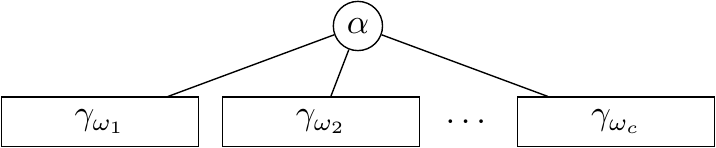}
\caption{The PQ-tree corresponding to the straight enumeration of a disconnected graph}
\label{fig:example PQ-tree disconnected}
\end{figure}

\medskip
We now indicate how to modify Algorithms \ref{alg:Robinsonian recognition} and \ref{alg:Robinsonian_main} in order to return a PQ-tree $\mathcal T$ encoding all the permutations ordering $A$ as a Robinson matrix.

We modify Algorithm \ref{alg:Robinsonian recognition} by taking as input, beside the matrix $A$ and the weak linear order $\psi$, also a node $\alpha$.
Then, the output is a PQ-tree $\mathcal{T}_{\alpha}$ rooted in $\alpha$, representing all the possible weak linear orders $\Phi$ compatible with $\psi$ and with straight enumerations of all the level graphs of $A$. It works as follows.

Let $G$ be the support of $A$. 
The idea is to recursively build a tree $\mathcal T_\omega$  for each  connected component $V_\omega$ of $G$ and then to merge these trees  according to the order of the components found by the routine  \textit{CO-Lex-BFS}($G,\psi$). 
To carry out this merging step we classify the components into the following three groups:
\begin{enumerate}
\item $\Theta$, which consists of all $\omega\in [c]$ for which the connected component $V_\omega$ meets at least two blocks of $\psi$.
\item $\Lambda$, which  consists of all $\omega\in [c]$ for which the component $V_\omega$ is contained in some block $B_i$, which 
 contains no other component. 
\item $\Omega=\cup_{i=1}^p\Omega_i$, where $\Omega_i$ consists of all $\omega\in [c]$ for which the component $V_\omega$ is contained in the block $B_i$,
which contains at least two components.
\end{enumerate}
Every time we  analyze a new connected component $\omega \in [c]$ in Algorithm \ref{alg:Robinsonian recognition}, we create a Q-node $\gamma_{\omega}$. 
After the common refinement $\Phi_{\omega}$  (of $\psi[V_\omega]$ and the straight enumeration $\phi_\omega$ of $G[V_\omega]$ or its reversal) has been computed, we have two possibilities.
If $A'[V_\omega]$ is diagonal, then we build the tree $\mathcal{T}_{\omega}$ rooted in $\gamma_{\omega}$ and whose children are P-nodes corresponding to the blocks of $\Phi_{\omega}$ (and prune the recursion tree  at this node).
Otherwise, we build the tree $\mathcal{T}_{\omega}$ recursively as output of \textit{Robinson}$(A'[V_\omega],\Phi_\omega,\gamma_\omega)$.\\ 
After all the connected components have been analyzed, we  insert the trees $\mathcal T_\omega$ in the final tree $\mathcal T_{\alpha}$ in the order they appear according to the routine \textit{CO-Lex-BFS}($G,\psi$). 
The root node is $\alpha$ and is given as input.
For each component $V_\omega$, we do the following operation to insert $\mathcal T_\omega$ in $\mathcal{T}_{\alpha}$, depending on the type of the component $V_\omega$:
\begin{enumerate}
\item If $\omega\in \Theta$, then $\phi_\omega$ (or $\overline{\phi}_{\omega}$) is the only straight enumeration compatible with $\psi[V_\omega]$.
Then we delete the node $\gamma_\omega$ and the children of $\gamma_\omega$ become children of $\alpha$ (in the same order). 
\item If $\omega\in \Lambda$, then both $\phi_\omega$ and its reversal $\overline{\phi}_\omega$ are compatible with $\psi[V_{\omega}]$. Then $\gamma_\omega$ becomes a child of $\alpha$.
\item If $\omega\in \Omega_i$ for some $i\in [p]$, then both $\phi_\omega$ and $\overline{\phi}_\omega$ are compatible with $\psi[V_\omega]$ and the same holds for any $\omega'\in \Omega_i$. Moreover, arbitrary permuting any two connected components $V_{\omega},V_{\omega'}$ with $\omega,\omega' \in \Omega_i$ will lead to a compatible straight enumeration. 
Then we insert a new node $\beta_i$ which is a P-node and  becomes a child of $\alpha$ and, for each $\omega'\in \Omega_i$,  $\gamma_{\omega'}$ becomes a child of $\beta_i$.
\end{enumerate}

\begin{center}
\resizebox{0.99\textwidth}{!}{
\begin{algorithm}[H] 
\SetKwInput {KwIn}{input}
\SetKwInput {KwOut}{output}
\KwIn{a nonnegative matrix $A \in \mathcal{S}^{n}$, a weak linear order $\psi$ of $V=[n]$ and a node $\alpha$}
\KwOut{A PQ-tree $\mathcal{T}_{\alpha}$ representing all the possible weak linear orders $\Phi$ compatible with $\psi$ and with straight enumerations of all the level graphs of $A$, or STOP (such a tree does not exist)}
\vspace{2ex}
$G$ is the support of A\\
\textit{CO-Lex-BFS}$(G,\psi)$ returns a linear order $(V_1,\dots,V_c)$ of the connected components of $G$ compatible with $\psi$ (if it exists) and a vertex order $\sigma$\\
group the connected components (c.c.) $V_{\omega}$ ($ \omega \in [c]$) of $G$ as follows:\\
$\Theta$ : all $\omega$ for which $V_\omega$ meets at least two blocks of $\psi$\\
$\Lambda$ : all $\omega$ for which  $V_\omega$ is contained in a block $B_i$ containing no other connected component\\
for $i\in [p]$,  $\Omega_i$: all $\omega$ for which  $V_\omega\subseteq B_i$ and $B_i$ contains at least two connected components\\
$\Phi = \emptyset$\\
\For {$\omega=1,\dots,c$}{
	create a Q-node $\gamma_{\omega}$\\
	$\phi_{\omega}$ = \textit{Straight\_enumeration}($G[V_{\omega}],\sigma[V_\omega]$) \hfill (if $G[V_{\omega}]$ is a unit int. graph)\\
	\If {$\Phi_{\omega}$ = \textit{Refine}$(\psi[V_{\omega}],\phi_{\omega}) = \emptyset$}
	{
		\If {$\Phi_{\omega}$ = \textit{Refine}$(\psi[V_{\omega}],\overline{\phi}_{\omega}) = \emptyset$}
		{
		\textbf{stop} \hfill (no straight enumeration compatible with $\psi[V_{\omega}]$ exists)
		}
	}
	$a'_{\min}$ is the smallest nonzero entry of $A[V_{\omega}]$\\
	$A'[V_{\omega}]$ is obtained from $A[V_{\omega}]$ by setting entries with value $a'_{\min}$ to zero\\
	\eIf{$A'[V_{\omega}]$ is diagonal}{
		create a PQ-tree $\mathcal{T}_{\omega}$ rooted in $\gamma_{\omega}$ and whose children are P-nodes corresponding to the blocks of $\Phi_{\omega}$			
	}{
		$\mathcal{T}_{\omega} = Robinson(A'[V_{\omega}],\Phi_{\omega},\gamma_{\omega})$
	}
}
$\mathcal{T}_{\alpha}$ is the PQ-tree rooted in $\alpha$, build as follows:\\
$\omega = 1$\\
\While {$\omega \leq c$}{
	\eIf{$\omega \in \Theta$}{
			the children of $\gamma_{\omega}$ become children of $\alpha$ and remove $\gamma_{\omega}$; $\omega=\omega+1$
		}{
		\eIf{$\omega \in \Lambda$}{
				set $\mathcal{T}_{{\omega}}$ as child of $\alpha$ (if $\alpha=\emptyset$, then set $\alpha=\gamma_{\omega}$); $\omega=\omega+1$ \\
		}{
		 	let $\Omega_i$ s.t. $\omega \in \Omega_i$; create a P-node $\beta_i$ and set it as child of $\alpha$ (if $\alpha=\emptyset$, then set $\alpha=\beta_i$)\\
		 	\ForEach {$\omega' \in \Omega_i$}{
		 		set $\gamma_{\omega'}$ as children of  $\beta_i$
		 	}
		 	$\omega= \omega + |\Omega_j|$		 	
		}{}
	}
}
\Return: $\mathcal{T}_{\alpha}$ 
\caption{\textit{\textit{Robinson}}$(A,\psi,\alpha)$}
\label{alg:PQ-tree}
\end{algorithm}
}
\end{center}

\newpage
Finally, we modify Algorithm \ref{alg:Robinsonian_main} by just giving the node $\alpha =\emptyset$ (i.e. undefined) as input to the first recursive call. 
The overall complexity of the algorithm after the above mentioned modifications is the same as for Algorithm \ref{alg:Robinsonian_main}. Indeed, determining the type of the connected components can be done in linear time, by just using the information about the initial and final blocks $B_{\omega}^{\min}$ and $B_{\omega}^{\max}$ already provided in Algorithm \ref{alg:Lex-BFSweak}. Furthermore, the operations on the PQ-tree are basic operations that do not increase the overall complexity of the algorithm.

\begin{algorithm}[!ht]
\SetKwInput {KwIn}{input}
\SetKwInput {KwOut}{output}
\KwIn{a nonnegative  matrix $A \in \mathcal{S}^{n}$}
\KwOut{a PQ-tree $\mathcal{T}$ that encodes all the permutations $\pi$ such that $A_{\pi}$ is a Robinson matrix or stating that $A$ is not Robinsonian}
\vspace{2ex}
$\psi=(V)$\\
$\alpha=\emptyset$\\
$G$ is the support of A\\
$\mathcal{T}$=\textit{Robinson}($A,\psi,\alpha$)\\
\eIf{the number of leaves of $\mathcal{T}$ is equal to $n$}{
		\Return: $\mathcal{T}$	
	}{
		``A is NOT Robinsonian"
}
\caption{\textit{\textit{Robinsonian}}$(A)$}
\label{alg:Robinsonian_mainPQ}
\end{algorithm}

\section{Conclusions}\label{secfinal} 

We introduced a new combinatorial algorithm to recognize Robinsonian matrices, based on a divide-and-conquer strategy and on a new characterization of Robinsonian matrices in terms of straight enumerations of unit interval graphs.
The algorithm is simple, rather intuitive and  relies only on basic routines like Lex-BFS and partition refinement, and it is well suited for sparse matrices.

The complexity depends on the depth $d$ of the recursion tree. An obvious bound on $d$ is the number $L$ of distinct entries in the matrix. A first natural question is to find other better bounds on the depth $d$. Is $d$ in the order $O(n)$, where $n$ is  the size of the matrix?
The answer is no: some computational experiments carried out in \cite{Seminaroti16} show that, for some instances, the depth of the recursion tree is $d = L > n$. This suggests that more sophisticated modifications might be needed to improve the complexity of the algorithm.
A possible way to bound the depth is to find criteria to prune recursion nodes. 
One possibility would be, when a submatrix is found for which the current weak linear order consists only of singletons, to check whether the corresponding permuted matrix is Robinson. 
 Another possible way to improve the complexity might be to compute the straight enumeration of the first level graph and then update it dynamically (in constant time, using a appropriate data structure)  without having to compute every time the whole straight enumeration of the next level graphs; this would need to extend the dynamic approach of \cite{Hell02}, which considers the case of single edge deletions, to the deletion of sets of edges.
Other possible future work includes investigating how  the algorithm could be used to design heuristics or approximation algorithm in the noisy case, when  $A$ is not Robinsonian, for example by using (linear) certifying algorithms as in \cite{Hell05}  to detect the edges and the nodes of the level graphs which create obstructions to being a unit interval graph. 

\section*{Acknowledgements}
This work was supported by the Marie Curie Initial Training Network ``Mixed Integer Nonlinear Optimization" (MINO)  grant no. 316647.

\bibliographystyle{plain}


\newpage

\appendix

\section{Example}\label{sec:appendix3}

We give a complete example of the algorithm.
We consider the  same matrix  $A$ as the one used in the example in Section 5 of \cite{Prea14}.
However, since \cite{Prea14} handles Robinsonian dissimilarities, we  first transform it into a similarity matrix and thus we use instead the matrix  $a_{\max} J - A$, where $a_{\max}$ denotes the largest entry in the matrix $A$.
If we rename such a new matrix as $A$, it looks as follows:

\scriptsize
\begin{equation*}
A =
\bordermatrix{
~ & \textcolor{red}{1} & \textcolor{red}{2} & \textcolor{red}{3} & \textcolor{red}{4} & \textcolor{red}{5} & \textcolor{red}{6} & \textcolor{red}{7} & \textcolor{red}{8} & \textcolor{red}{9} & \textcolor{red}{10}& \textcolor{red}{11} & \textcolor{red}{12} & \textcolor{red}{13} & \textcolor{red}{14} & \textcolor{red}{15} & \textcolor{red}{16} & \textcolor{red}{17} & \textcolor{red}{18} & \textcolor{red}{19}\cr
\textcolor{red}{1}  & 11 & 2 & 9 & 0 & 5 & 0 & 5 & 5 & 2 & 0 & 5 & 0 & 5 & 6 & 0 & 0 & 2 & 0 & 5\cr
\textcolor{red}{2}  & & 11 & 2 & 0 & 9 & 0 & 8 & 5 & 10 & 0 & 5 & 0 & 5 & 2 & 0 & 0 & 10 & 0 & 8\cr
\textcolor{red}{3}  & & & 11 & 0 & 5 & 0 & 5 & 5 & 2 & 0 & 5 & 0 & 5 & 10 & 0 & 0 & 2 & 0 & 5\cr
\textcolor{red}{4}  & & & & 11 & 0 & 3 & 0 & 0 & 0 & 3 & 0 & 3 & 0 & 0 & 10 & 3 & 0 & 9 & 0\cr
\textcolor{red}{5}  & & & & & 11 & 0 & 8 & 7 & 9 & 0 & 7 & 0 & 7 & 5 & 0 & 0 & 9 & 0 & 10\cr
\textcolor{red}{6}  & & & & & & 11 & 0 & 0 & 0 & 10 & 0 & 6 & 0 & 0 & 5 & 8 & 0 & 5 & 0\cr
\textcolor{red}{7}  & & & & & & & 11 & 7 & 8 & 0 & 7 & 0 & 7 & 5 & 0 & 0 & 8 & 0 & 9\cr
\textcolor{red}{8}  & & & & & & & & 11 & 6 & 0 & 10 & 0 & 8 & 7 & 0 & 0 & 6 & 0 & 7\cr
\textcolor{red}{9}  & & & & & & & & & 11 & 0 & 6 & 0 & 5 & 2 & 0 & 0 & 10 & 0 & 8\cr
\textcolor{red}{10} & & & & & & & & & & 11 & 0 & 6 & 0 & 0 & 4 & 9 & 0 & 5 & 0\cr
\textcolor{red}{11} & & & & & & & & & & & 11 & 0 & 9 & 7 & 0 & 0 & 6 & 0 & 7\cr
\textcolor{red}{12} & & & & & & & & & & & & 11 & 0 & 0 & 9 & 6 & 0 & 10 & 0\cr
\textcolor{red}{13} & & & & & & & & & & & & & 11 & 7 & 0 & 0 & 5 & 0 & 7\cr
\textcolor{red}{14} & & & & & & & & & & & & & & 11 & 0 & 0 & 2 & 0 & 5\cr
\textcolor{red}{15} & & & & & & & & & & & & & & & 11 & 4 & 0 & 10 & 0\cr
\textcolor{red}{16} & & & & & & & & & & & & & & & & 11 & 0 & 4 & 0\cr
\textcolor{red}{17} & & & & & & & & & & & & & & & & & 11 & 0 & 8\cr
\textcolor{red}{18} & & & & & & & & & & & & & & & & & & 11 & 0\cr
\textcolor{red}{19} & & & & & & & & & & & & & & & & & & & 11\cr
}
\end{equation*}
\normalsize
Here the red labels denote the original numbering of the elements. Throughout we will use the fact that adding any multiple of the all-ones matrix $J$ to the matrix $A$ does not change the Robinson(ian) property. 
The recursion tree computed by Algorithm \ref{alg:Robinsonian_main} is shown in Figure \ref{fig:Recursion tree} at page \pageref{fig:Recursion tree}.
The weak linear order at each node represents the weak linear order $\psi$ given as input to the recursion node, while the number on the edge between two nodes denotes the minimum value in the current matrix $A$, which is  set to zero before making a new recursion call (in this way, the reader may reconstruct the input given at each recursion node).

\begin{sidewaysfigure}[!ht]
\centering
\includegraphics[width=\linewidth]{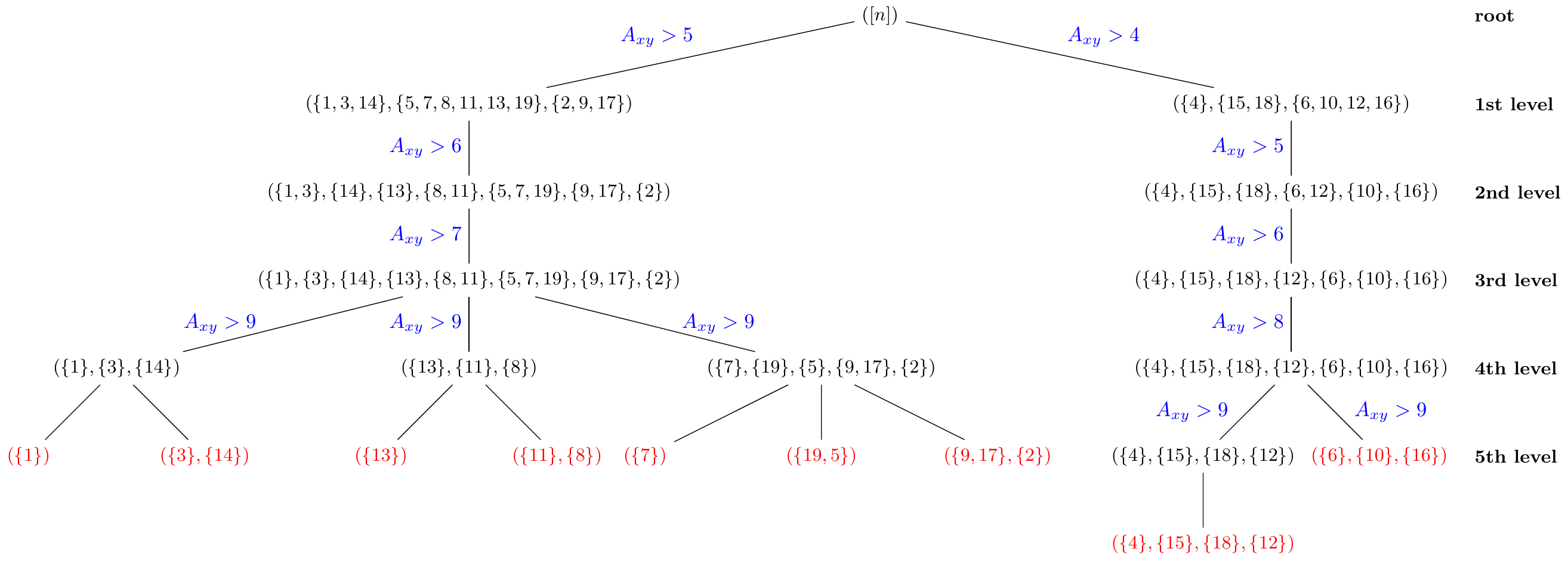}
\caption{Recursion tree. The weak linear order at each node represent the weak linear order $\psi$ given as input to the recursion node, while the number on the edge between two nodes denotes the minimum value in $A$ set to zero before making a new recursion call. The red weak linear orders are the pruned nodes.}
\label{fig:Recursion tree}
\end{sidewaysfigure}

\medskip \noindent
\textbf{Root node}\\
We set $\psi =(V)$ and invoke Algorithm \ref{alg:Robinsonian recognition}.
Then, Algorithm \ref{alg:Lex-BFSweak} would find two connected components:
\begin{align*}
& V_{1} = \{1,2,3,5,7,8,9,11,13,14,17,19\}, &\\
& V_{2} = \{4,6,10,12,15,16,18\}. &
\end{align*}
\noindent
Hence, we can split the problem into two subproblems, where we deal with each connected component independently.

\medskip \noindent
\textbf{1.0 Connected component $V_1$, level 0}\\
The submatrix $A[V_1]$ induced by $V_1$ is shown below (after shifting the matrix by $a_{\min}=2$, i.e., substracting $2J$ from it). 

\begin{equation}\label{eq:A_1}
A[V_1]=
\bordermatrix{
~ & \textcolor{red}{1} & \textcolor{red}{2} & \textcolor{red}{3} & \textcolor{red}{5} & \textcolor{red}{7} & \textcolor{red}{8} & \textcolor{red}{9} & \textcolor{red}{11} & \textcolor{red}{13} & \textcolor{red}{14}& \textcolor{red}{17} & \textcolor{red}{19}\cr
\textcolor{red}{1}  & 9 & 0 & 7 & 3 & 3 & 3 & 0 & 3 & 3 & 4 & 0 & 3 \cr
\textcolor{red}{2}  & & 9 & 0 & 7 & 6 & 3 & 8 & 3 & 3 & 0 & 8 & 6 \cr
\textcolor{red}{3}  & & & 9 & 3 & 3 & 3 & 0 & 3 & 3 & 8 & 0 & 3 \cr
\textcolor{red}{5}  & & & & 9 & 6 & 5 & 7 & 5 & 5 & 3 & 7 & 8 \cr
\textcolor{red}{7}  & & & & & 9 & 5 & 6 & 5 & 5 & 3 & 6 & 7 \cr
\textcolor{red}{8}  & & & & & & 9 & 4 & 8 & 6 & 5 & 4 & 5 \cr
\textcolor{red}{9}  & & & & & & & 9 & 4 & 3 & 0 & 8 & 6 \cr
\textcolor{red}{11} & & & & & & & & 9 & 7 & 5 & 4 & 5 \cr
\textcolor{red}{13} & & & & & & & & & 9 & 5 & 3 & 5 \cr
\textcolor{red}{14} & & & & & & & & & & 9 & 0 & 3 \cr
\textcolor{red}{17} & & & & & & & & & & & 9 & 6 \cr
\textcolor{red}{19} & & & & & & & & & & & & 9 \cr
}
\end{equation}

\noindent
If we invoke Algorithm \ref{alg:straght_enumeration} on the support $G[V_1]$, we get the following straight enumeration:
\begin{equation*}
\phi = (\{1,3,14\},\{5,7,8,11,13,19\},\{2,9,17\}).
\end{equation*}
Note that since $\psi[V_1]$ has only one block, we do not need to compute the partition refinement in Algorithm \ref{alg:refine} and then the common refinement is simply $\Phi=\phi$.
The smallest nonzero value of the matrix in (\ref{eq:A_1}) is $a'_{\min} = 3$. Hence, we compute $A'[V_1]$ by setting to zero the entries of the matrix in (\ref{eq:A_1}) with value equal to $a'_{\min}$.
Since $A'[V_1]$ is not diagonal, we make a recursion call, and we set $\psi=\Phi$. To simplify notation, we shall rename $A'[V_1]$ as $A[V_1]$ after every iteration.

\medskip
\noindent
\textbf{1.1 Connected component $V_1$, level 1}\\
The input matrix $A[V_1]$ is obtained by setting to zero the entries of the matrix in (\ref{eq:A_1}) with value at most 3. The support of this matrix is still connected, and its straight enumeration is:
\begin{equation*}
\phi = (\{1,3\},\{14\},\{13\},\{8,11\},\{5,7,19\},\{9,17\},\{2\}).
\end{equation*}
If we invoke Algorithm \ref{alg:refine}, it is easy to see that the common refinement $\Phi$ of $\psi$ and $\phi$ is exactly $\phi$.
The smallest nonzero value of $A[V_1]$ is now $a'_{\min} = 4$. Hence, we compute $A'[V_1]$, which is not diagonal and thus we make a recursion call, setting $\psi=\Phi$ and renaming $A'[V_1]$ as $A[V_1]$.

\medskip
\noindent
\textbf{1.2 Connected component $V_1$, level 2}\\
The input matrix $A[V_1]$ is obtained by setting to zero the entries of the matrix in (\ref{eq:A_1}) with value at most 4.
The support of this matrix is still connected, and its straight enumeration is:
\begin{equation*}
\phi = (\{1\},\{3\},\{14\},\{13,8,11\},\{5,7,19\},\{9,17,2\}).
\end{equation*}
The common refinement with $\psi$ is then given by:
\begin{equation*}
\Phi = (\{1\},\{3\},\{14\},\{13\},\{8,11\},\{5,7,19\},\{9,17\},\{2\}).
\end{equation*}
The smallest nonzero value of $A[V_1]$ is now $a'_{\min} = 5$. Hence, we compute $A'[V_1]$, which is not diagonal and thus we make a recursion call, setting $\psi=\Phi$ and renaming $A'[V_1]$ as $A[V_1]$.

\medskip
\noindent
\textbf{1.3 Connected component $V_1$, level 3}\\
The input matrix $A[V_1]$ is obtained by setting to zero the entries of the matrix in (\ref{eq:A_1}) with value at most 5.
The support of this matrix is not connected, and thus Algorithm \ref{alg:Lex-BFSweak} will detect the following connected components:
\begin{align*}
& V_{11} = \{1,3,14\}, &\\
& V_{12} = \{13,8,11\}, &\\
& V_{13} = \{5,7,19,9,17,2\}. &
\end{align*}
We analyze each connected component independently.

\medskip
\noindent
\textbf{1.3.1 Connected component $V_{11}$, level 4}\\
The common refinement is $\Phi = (\{1\},\{3\},\{14\})$ and $a'_{\min}=7$. We make a recursion call, finding two connected components $\{1\}$ and $\{3,14\}$, and then we stop, because the first one has only one vertex, while the submatrix corresponding to the second one is diagonal (after shifting). 

\medskip
\noindent
\textbf{1.3.2 Connected component $V_{12}$, level 4}\\
Since $a_{\min}=6$, we first ``shift" the input submatrix. Then, the common refinement is $\Phi = (\{13\},\{11\},\{8\})$ and $a'_{\min}=7$ (i.e. =1 after shifting). We make a recursion call, finding two connected components $\{13\}$ and $\{11,8\}$, and then we stop, because the first one has only one vertex, while the submatrix corresponding to the second one is diagonal (after shifting). 

\medskip
\noindent
\textbf{1.3.3 Connected component $V_{13}$, level 4}\\
Since $a_{\min}=6$, we first ``shift" the input submatrix. Then, the common refinement is $\Phi = (\{7\},\{19\},\{5\},\{9,17\},\{2\})$ and $a'_{\min}=7$ (i.e. =1 after shifting).
We update $A'[V_1]$, and we make a recursive call because it is not diagonal.
The new input matrix is then given by the submatrix in (\ref{eq:A_1}) restricted to $V_{13}$ by setting to zero the entries with value at most 7.
The support of this  matrix is not connected, and thus Algorithm \ref{alg:Lex-BFSweak} will detect the following connected components:
\begin{align*}
& V_{131} = \{7\}, &\\
& V_{132} = \{19,5\}, &\\
& V_{133} = \{9,17,2\}. &
\end{align*}
\noindent
We then split the problem over the connected components. The first one has only one vertex while the second and the third one are diagonal (after shifting). 
This was the last recursion node open. 

\medskip\noindent
Therefore, we get that the final common refinement of the level graphs of the matrix in (\ref{eq:A_1}) is:
\begin{equation*}
\Phi_{1} = (\{1\},\{3\},\{14\},\{13\},\{11\},\{8\},\{7\},\{19\},\{5\},\{9,17\},\{2\})
\end{equation*}
and the PQ-tree $\mathcal{T}_{1}$ computed by the algorithm is reported in Figure \ref{fig:PQ_example_1_7} at 
page~\pageref{fig:PQ_example_1_7}.
\begin{figure}[!h]
\centering
  	\includegraphics[width=.7\linewidth]{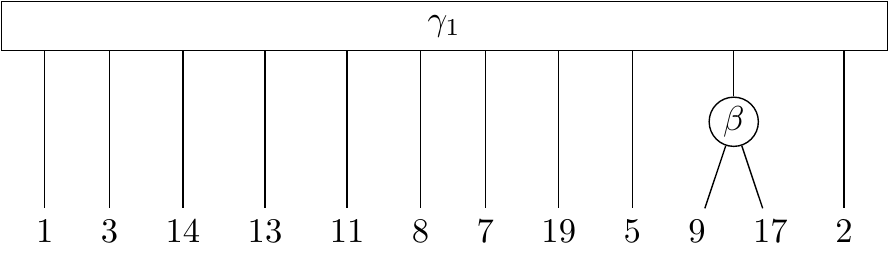}
\caption{The PQ-tree corresponding to the common refinement of level graphs of $A[V_1]$}
\label{fig:PQ_example_1_7}
\end{figure}

\medskip \noindent
\textbf{2.0 Connected component $V_2$, level 0}\\
The submatrix $A[V_2]$ induced by $V_2$ is reported below (after shifting the matrix by $a_{\min}=3$). 

\begin{equation}\label{eq:A_2}
A[V_2]=
\bordermatrix{
~ & \textcolor{red}{4} & \textcolor{red}{6} & \textcolor{red}{10} & \textcolor{red}{12} & \textcolor{red}{15} & \textcolor{red}{16} & \textcolor{red}{18} \cr
\textcolor{red}{4}  & 8 & 0 & 0 & 0 & 7 & 0 & 6 \cr
\textcolor{red}{6}  & & 8 & 7 & 3 & 2 & 5 & 2 \cr
\textcolor{red}{10}  & & & 8 & 3 & 1 & 6 & 2 \cr
\textcolor{red}{12}  & & & & 8 & 6 & 3 & 7 \cr
\textcolor{red}{15}  & & & & & 8 & 1 & 7 \cr
\textcolor{red}{16}  & & & & & & 8 & 1 \cr
\textcolor{red}{18}  & & & & & & & 8 \cr
}
\end{equation}

\noindent
If we invoke Algorithm \ref{alg:straght_enumeration} on the support $G[V_2]$, we get the following straight enumeration:
\begin{equation*}
\phi =(\{4\},\{15,18\},\{6,10,12,16\}).
\end{equation*}
Note that since $\psi[V_{2}]$ has only one block, we do not have to compute the partition refinement in Algorithm \ref{alg:refine}, and then the common refinement is simply $\Phi=\phi$.
The smallest nonzero value of the matrix in (\ref{eq:A_2}) is $a'_{\min} = 1$. Hence, we compute $A'[V_2]$ by setting to zero the entries of the matrix in (\ref{eq:A_2}) with value equal to $a'_{\min}$.
Since $A'[V_2]$ is not diagonal, we make a recursion call, and we set $\psi=\Phi$. To simplify notation, we shall again rename $A'[V_2]$ as $A[V_2]$ after every iteration.

\medskip
\noindent
\textbf{2.1 Connected component $V_2$, level 1}\\
The input matrix $A[V_2]$ is obtained by setting to zero the entries of the matrix in (\ref{eq:A_2}) with value at most 1.
The support of this matrix is still connected, and its straight enumeration is:
\begin{equation*}
\phi=(\{4\},\{15\},\{18\},\{6,12\},\{10\},\{16\}).
\end{equation*}
If we invoke Algorithm \ref{alg:refine}, it is easy to see that the common refinement is $\Phi = \phi$
The smallest nonzero value of $A[V_2]$ is now $a'_{\min} = 2$. Hence, we compute $A'[V_2]$, which is not diagonal and thus we make a recursion call, setting $\psi=\Phi$ and renaming $A'[V_2]$ as $A[V_2]$ for the next iteration.

\medskip
\noindent
\textbf{2.2 Connected component $V_2$, level 2}\\
The input matrix $A[V_2]$ is obtained by setting to zero the entries of the matrix in (\ref{eq:A_2}) with value at most 2.
The support of this matrix is still connected, and its straight enumeration is:
\begin{equation*}
\phi=(\{4,15\},\{18\},\{12\},\{6,10,16\}).
\end{equation*}
The common refinement is then simply:
\begin{equation*}
\Phi = (\{4\},\{15\},\{18\},\{12\},\{6\},\{10\},\{16\}).
\end{equation*}
The smallest nonzero value of $A[V_2]$ is now $a'_{\min} = 3$. Hence, we compute $A'[V_2]$, which is not diagonal and thus we make a recursion call, setting $\psi=\Phi$ and renaming $A'[V_2]$ as $A[V_2]$ for the next iteration.

\medskip
\noindent
\textbf{2.3 Connected component $V_2$, level 3}\\
The input matrix $A[V_2]$ is obtained by setting to zero the entries of the matrix in (\ref{eq:A_2}) with value at most 3.
The support of this matrix is still connected, and its straight enumeration is:
\begin{equation*}
\phi=(\{4\},\{15\},\{18\},\{12\},\{6\},\{10\},\{16\})
\end{equation*}
which is equal to the given $\psi$ (and thus will be the common refinement $\Phi$).
The smallest nonzero value of $A[V_2]$ is now $a'_{\min} = 5$. Hence, we compute $A'[V_2]$, which is not diagonal and thus we make a recursion call, setting $\psi=\Phi$ and renaming $A'[V_2]$ as $A[V_2]$ for the next iteration.

\medskip
\noindent
\textbf{2.4 Connected component $V_2$, level 4}\\
The input matrix $A[V_2]$ is obtained by setting to zero the entries of the matrix in (\ref{eq:A_2}) with value at most 5. 
The support of this matrix is not connected, and thus Algorithm \ref{alg:Lex-BFSweak} will detect two connected components.
\begin{align*}
& V_{21} = \{4,15,18,12\}, &\\
& V_{22} = \{6,10,16\}. &
\end{align*}
We then split the problem over the connected components.

\medskip
\noindent
\textbf{2.4.1 Connected component $V_{21}$, level 5}\\
The common refinement is $\Phi = (\{4\},\{15\},\{18\},\{12\})$ and $a'_{\min}=6$. We then make a recursion call and
we find again a connected graph. Again the common refinement does not change (in fact the blocks are singletons) and now $a'_{\min}=7$. Finally, the submatrix $A'[V_{21}]$ is diagonal, so we prune the node.

\medskip
\noindent
\textbf{2.4.2 Connected component $V_{22}$, level 5}\\
The common refinement is $\Phi = (\{6\},\{10\},\{16\})$ (again the blocks are singletons) and $a'_{\min}=6$. But now $A'[V_{22}]$ is diagonal, so we prune the node.
This was the last recursion node open.

\medskip\noindent
Therefore, we get that the final common refinement of level graphs of $A[V_2]$ is:
\begin{equation*}
\Phi_2 = (\{4\},\{15\},\{18\},\{12\},\{6\},\{10\},\{16\})
\end{equation*}
and the PQ-tree $\mathcal{T}_{2}$ computed by the algorithm is shown in Figure \ref{fig:PQ_example_2_3} at page \pageref{fig:PQ_example_2_3}.
\begin{figure}[!h]
\centering
  	\includegraphics[width=.4\linewidth]{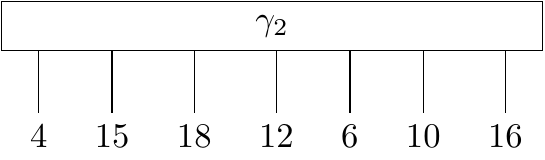}
\caption{The PQ-tree corresponding to the common refinement of level graphs of $A[V_2]$}
\label{fig:PQ_example_2_3}
\end{figure}
Finally, we can build the PQ-tree representing the permutation reordering $A$ as a Robinson matrix. Since both $V_1$ and $V_2$ are contained in the same block of $\psi$ (which at the beginning is $\psi=([n])$), then we create a P-node (named $\alpha$ since it is the ancestor) whose children are the subtrees $\mathcal{T}_{1}$ and $\mathcal{T}_{2}$.
The final PQ-tree is shown in Figure \ref{fig:PQ_example_3_final} at page \pageref{fig:PQ_example_3_final}, and is equivalent to the one returned by \cite{Prea14}.

\begin{figure}[!h]
\centering
  	\includegraphics[width=\linewidth]{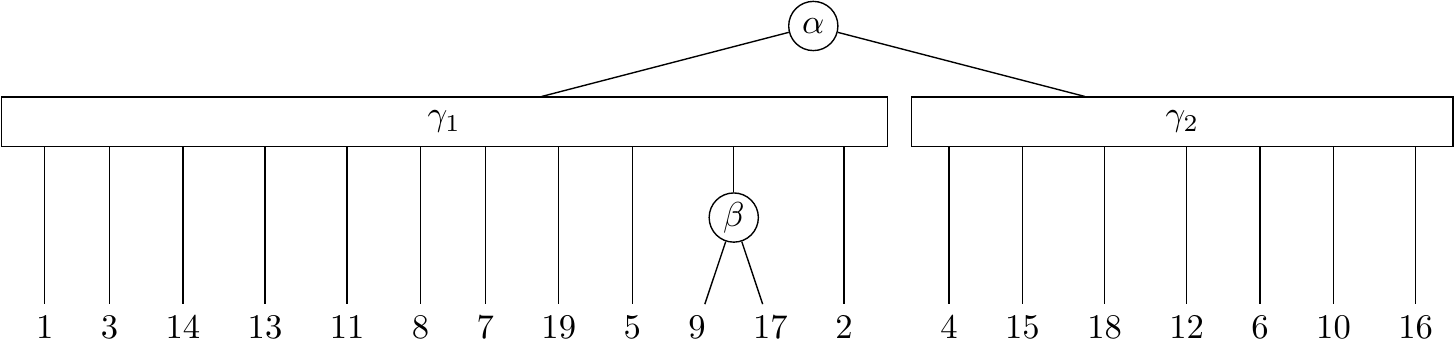}
\caption{The PQ-tree corresponding to the permutations reordering $A$ as a Robinson matrix}
\label{fig:PQ_example_3_final}
\end{figure}

Note that, using the fact that a common refinement which is a linear order cannot be refined anymore, the depth could be lowered to $d=4$ (since the right branch would have depth $d=3$).
Understanding this and other possible speed up  will be the subject of future work.

\end{document}